\documentclass[11pt]{article}
\usepackage{geometry}                
\usepackage{graphicx}
\usepackage{epsfig}
\usepackage{color}
\usepackage{subfigure}
\usepackage{graphics}

\usepackage{amssymb}
\usepackage{amsthm}
\usepackage{amsmath}
\usepackage{amscd}
\usepackage{epstopdf}

\usepackage[ruled,longend]{algorithm2e}

\usepackage{fullpage}

\usepackage{authblk}

\usepackage[ruled]{algorithm2e}

\newtheorem{theorem}{Theorem}[section]
\newtheorem{lemma}[theorem]{Lemma}

\newtheorem{corollary}[theorem]{Corollary}

\newtheorem{definition}[theorem]{Definition}

\newtheorem{invariant}[theorem]{Invariant}
\newtheorem*{inductionHypothesis}{Induction Hypothesis}
\newtheorem*{baseCase}{Base Case}

\newtheorem*{notation}{Notation}

\newtheorem{remark}[theorem]{Remark}

\title{Simpler, Linear-time Transitive Orientation via Lexicographic Breadth-First Search}

\author{Marc Tedder\thanks{The author wishes to thank NSERC for its financial support throughout the research culminating in this paper.
Email: \texttt{mtedder@cs.toronto.edu}}} 
\affil{University of Toronto}

\date{}

\begin{document}

\maketitle

\begin{abstract}
Comparability graphs are the undirected graphs whose edges can be directed so that the resulting directed graph is transitive.  They are related to posets and have applications in scheduling theory.  This paper considers the problem of finding a transitive orientation of a comparability graph, a requirement for many of its applications.  A linear-time algorithm is presented based on an elegant partition refinement scheme developed elsewhere for the problem.  The algorithm is intended as a simpler and more practical alternative to the existing linear-time solution, which is commonly understood to be difficult and mainly of theoretical value.  It accomplishes this by using Lexicographic Breadth-First Search to achieve the same effect as produced by modular decomposition in the earlier linear-time algorithm.  
\end{abstract}

\section{Introduction}

\emph{Comparability graphs} are a well-studied class of graphs related to posets.  Every poset can be represented by a transitive, directed, acyclic graph; comparability graphs are formed by ignoring the directions on these edges.  Equivalently, comparability graphs are those undirected graphs whose edges can be directed so that the resulting directed graph is transitive: whenever there is an edge from $a$ to $b$ and one from $b$ to $c$, there is also an edge from $a$ to $c$.  Given a comparability graph, the \emph{transitive orientation problem} asks you to find such a transitive orientation of its edges.  An example of a comparability graph is provided in figure~\ref{fig:example}, while the same graph with its edges transitively oriented appears in figure~\ref{fig:transitiveOrientation}.    Graphs whose complement is a comparability graph are called \emph{co-comparability} graphs.  A related problem is to compute a transitive orientation of the complement of a co-comparability graph.  \emph{Permutation graphs} are those that are both comparability graphs and co-comparability graphs; these correspond to partial orders of dimension-2.  (See~\cite{Brandstadt:1999:GCS:302970} for a summary of these graph classes.)

Many interesting combinatorial problems can be efficiently solved for comparability graphs and co-comparability graphs (and thus permutation graphs) once a transitive orientation is known~\cite{m.mcconnell:linear-time}.  This is the source of their many applications in scheduling theory~\cite{Mohring:scheduling,Korte:scheduling}.  Earlier solutions to the transitive orientation problem ran in time $O(n^2)$~\cite{Spinrad:nsquared} and $O(\delta{}m)$~\cite{Golumbic:transitive,Ghouila-Houri:transitive,Pnueli:transitive}, where $\delta$ is the maximum degree of a vertex in the graph.  The approach developed in the $O(n^2)$ algorithm of~\cite{Spinrad:nsquared} was later extended to $O(m\log{n})$~\cite{m.mcconnell:partitioning} and $O(n+m)$~\cite{m.mcconnell:linearTransitive} algorithms.  While the former presents an elegant solution that is straightforward in its understanding and implementation, linear-time in the latter was achieved at the expense of such simplicity.  Although the latter's algorithm's difficulty has been acknowledged by the authors in~\cite{m.mcconnell:partitioning,spinrad:book}, it has so far remained the only linear-time solution to the transitive orientation problem.  Part of the acknowledged difficulty is its reliance on the modular decomposition algorithm also developed in~\cite{m.mcconnell:linearTransitive}.  

Modular decomposition has long been known to simplify the computation of transitive orientations~\cite{gallai:transitive} in addition to numerous other problems~\cite{mohring:applications}.  The algorithm in~\cite{m.mcconnell:linear-time} and another independently developed in~\cite{habib:modular} were the first linear-time modular decomposition algorithms.  Unfortunately, both these algorithms were complex enough to be deemed mainly theoretical contributions~\cite{bretscher:lexbfs}.  But the practical importance of modular decomposition led many to attempt simpler, linear-time alternatives~\cite{habib:stacs,dahlhaus:modular} (note that while~\cite{habib:swat} claimed linear-time, the authors have since noted an error in the paper~\cite{tedder:communication}), something finally achieved in~\cite{tedder:modular} by unifying earlier approaches.  Similar attempts at a simpler, linear-time alternative to the transitive orientation algorithm of~\cite{m.mcconnell:linearTransitive} have so far not appeared.

The elegant $O(m\log{n})$ transitive orientation algorithm of~\cite{m.mcconnell:partitioning} is based on partition refinement techniques (see~\cite{habib:stacs} for an overview).  Some hope was expressed there that a simpler linear-time implementation of these techniques may be possible.  This paper posits an algorithm that proposes to answer that question in the affirmative.  Where previous algorithms used modular decomposition, this paper achieves the same effect with Lexicographic Breadth-First Search (LBFS) (see~\cite{corneil:lexbfs} for a survey).  The key point is that LBFS is itself easily implemented in linear-time using partition refinement~\cite{habib:stacs}.  The result is a simpler, linear-time alternative to solving the transitive orientation problem.  The paper also suggests how its approach with LBFS can be adapted to compute a transitive orientation of the complement of a co-comparability graph in linear-time, which necessarily avoids explicitly computing that complement.  All of this can be implemented using partition refinement techniques that amount to basic list and tree traversals.

\section{Background and Overview} \label{sec:background}

\subsection{Preliminaries}

Only simple graphs are considered in this paper.  All graphs should be assumed to be undirected unless otherwise specified.  The set of vertices of a graph $G$ will be denoted $V(G)$ and its set of edges $E(G)$.  Throughout this paper $n$ will be used to refer to $|V(G)|$ and $m$ for $|E|$.  The graph induced by the set of vertices $S \subseteq V(G)$ will be denoted $G[S]$.  The set of neighbours of a vertex $x$ will be denoted $N(x)$, while $N[x] = N(x) \cup \{x\}$.  Depending on the context, $xy$ will either represent the undirected edge $\{x,y\}$ or the directed edge $(x,y)$.  The connected components of a graph will simply be referred to as its \emph{components}; the connected components of the complement of a graph will be called its \emph{co-components}.         

A vertex $x$ is said to be \emph{universal} to a set of vertices $S$ if $x$ is adjacent to every vertex in $S$, and $x$ is \emph{isolated} from $S$ if it is adjacent to no vertex in $S$.  A set $S$ is said to be \emph{universal} to another set $S'$ if every vertex in $S$ is adjacent to every vertex in $S'$ (i.e. there is a \emph{join} between the two sets).  On the other hand, $x$ \emph{splits} $S$ if it is adjacent to at least one but not every vertex in $S$, and a set $S$ \emph{splits} another set $S'$ if $S \cap S' \ne \emptyset$ and $S' - S \ne \emptyset$.  Meanwhile, two sets $S$ and $S'$ are said to \emph{overlap} if $S - S'$, $S \cap S'$, and $S' - S$ are all non-empty.  A set $S \subseteq V(G)$ is \emph{trivial} if $S = V(G)$ or $|S| = 1$.  

The vertices of a rooted tree will be referred to as \emph{nodes}, with the non-leaf vertices specially designated as \emph{interior nodes}.  An \emph{ancestor} of a node $x$ is a vertex other than $x$ on the path from that node to the root.  A node $x$ is a \emph{descendant} of the node $y \ne x$ if $y$ appears on the path from $x$ to the root.  

\subsection{Comparability Graphs} 

A directed graph is transitive if whenever there are edges $ab$ and $bc$ there is also the edge $ac$.  Vertices whose incident edges are all directed the same way are either \emph{sources} (directed outward) or \emph{sinks} (directed inward).  An undirected graph whose edges can be transitively oriented is a \emph{comparability graph}.  A \emph{linear extension} of a comparability graph $G$ is an ordering  of $V(G)$ that induces the following transitive orientation: the edge $ab$ is directed toward $b$ if and only if $a$ appears before $b$ in the ordering.  An example of a comparability graph appears in figure~\ref{fig:example}; while an example of a transitive orientation of its edges appears in figure~\ref{fig:transitiveOrientation}.  Observe that $x,z,q,w,r,v,y,u,a,b$ is a linear extension for this graph.    

\begin{figure}  
\begin{center}  
\includegraphics[scale=0.85]{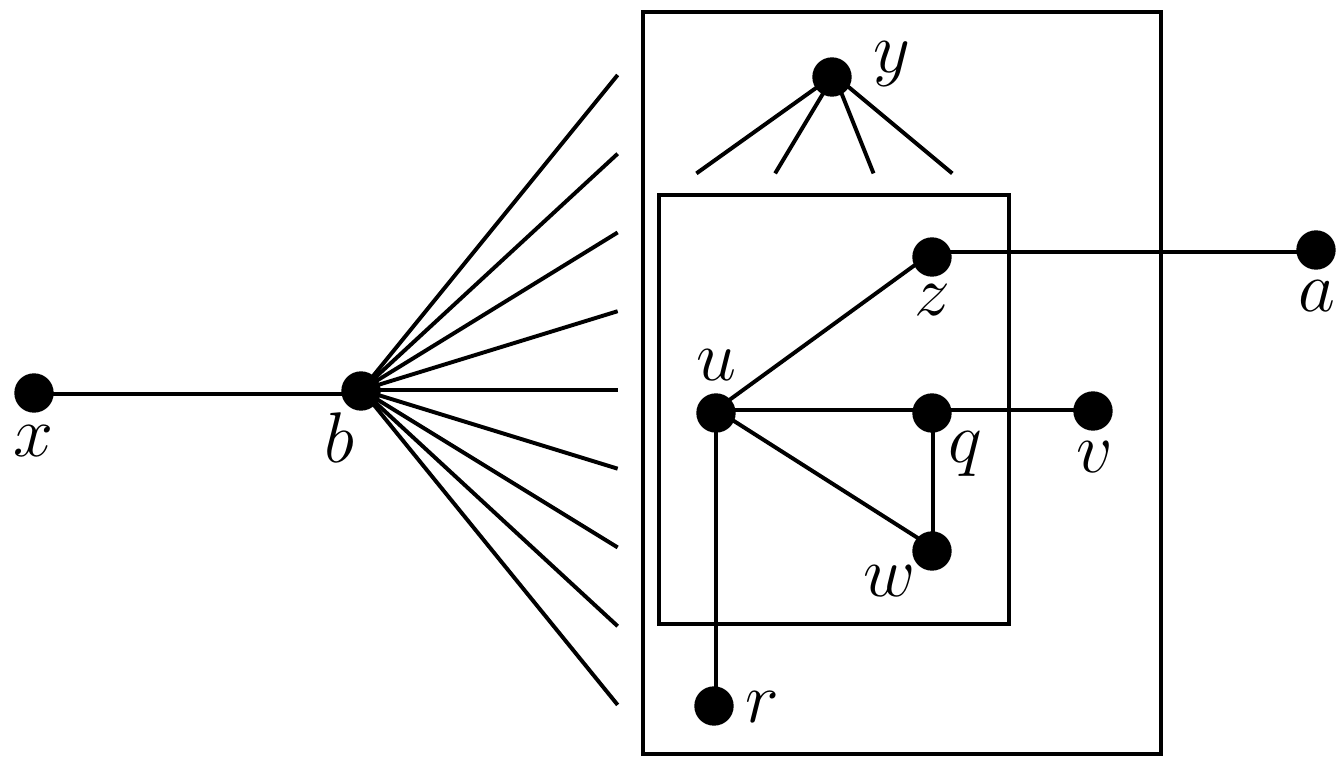}
\caption{A comparability graph -- vertices are grouped in boxes, and edges from one vertex to the perimeter of a box means that vertex is universal to the vertices inside the box.}
\label{fig:example}
\end{center}   
\end{figure}

\begin{figure}  
\begin{center}  
\includegraphics[scale=0.85]{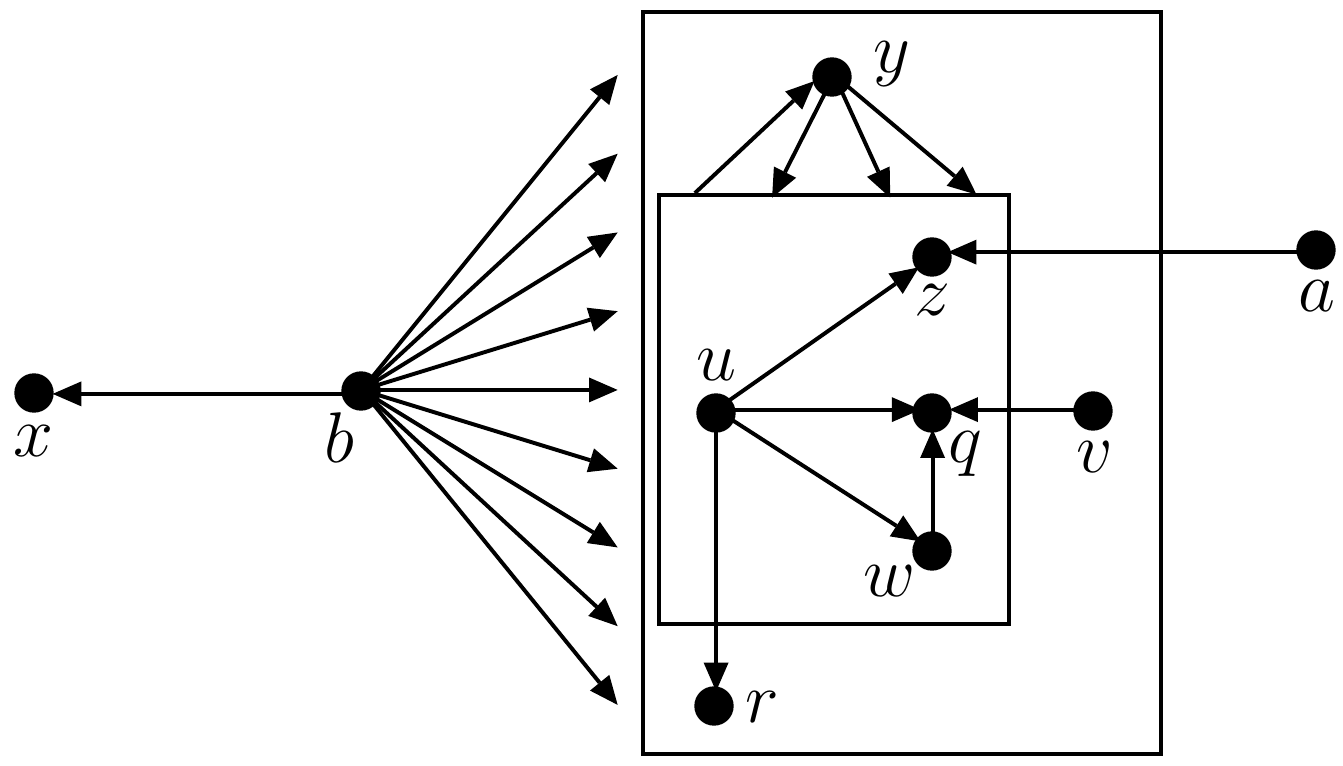}
\caption{The comparability graph of figure~\ref{fig:example} with its edges transitively oriented.}
\label{fig:transitiveOrientation}
\end{center}   
\end{figure}

\subsection{Modular Decomposition}

A \emph{module} is a set of vertices $M$ such that no vertex $x \notin M$ splits it.  A module $M$ is \emph{strong} if it overlaps no other module.  A graph is \emph{prime} if it contains no non-trivial module.   Note that the graph in figure~\ref{fig:example} is prime.  The \emph{modular decomposition tree} for the graph $G$ is the rooted tree defined as follows.

\begin{itemize}

\item The set of leaves corresponds to the set of vertices $V(G)$.

\item The set of interior nodes corresponds to the set of strong modules.

\item There is an edge between two nodes if one is the smallest strong module containing the other.

\end{itemize}

Every interior node in the modular decomposition tree defines a \emph{quotient graph}: the graph induced by the set of vertices formed by selecting a single leaf descending from each of its children.  It can be shown that the quotient graph is invariant under different choices of leaves in that the resulting graphs are all isomorphic.  Every quotient graph is either a clique, independent set, or is prime (see, for example,~\cite{habib:modular}).  

Modules play an important role in the transitive orientation of comparability graphs.  Given a module $M$ in a comparability graph $G$, and a vertex $x \in M$, it suffices to separately compute a transitive orientation for $G[M]$ and $G[(V(G) - M) \cup \{x\}]$: merely replace $x$ in the latter with $M$.  This can be applied to the modular decomposition tree's strong modules and quotient graphs.  Since cliques and independent sets can trivially be oriented, the problem therefore reduces to the prime case.  This is made somewhat easier by the fact that prime comparability graphs have a unique transitive orientation (up to reversal of the directions on all edges)~\cite{gallai:transitive}.

\subsection{Partition Refinement} \label{sec:partitionRefinement}

A \emph{partition} of the set of vertices $V(G)$ is a collection of disjoint subsets $\mathcal{P} = \{P_1,\ldots,P_k\}$ such that $V(G) = P_1 \cup \cdots \cup P_k$.  Each $P_i$ is called a \emph{partition class}.  A \emph{refinement} of the partition $\mathcal{P}$ is one in which every partition class is a (not necessarily proper) subset of a partition class in $\mathcal{P}$.  The partition of $S \subseteq V(G)$ induced by $\mathcal{P}$, denoted $\mathcal{P}[S]$, is obtained by taking the intersection of $S$ with every partition class in $\mathcal{P}$ (ignoring any empty intersections).  Partition refinement frequently operates on \emph{ordered partitions}: those having an ordering imposed on their partition classes, written $\mathcal{P} = P_1,\ldots,P_k$.  

Refinements in this paper will be computed according to the pivot operation defined by algorithm~\ref{alg:pivot}.  The input vertex $p$ is referred to as a  \emph{pivot}.  Algorithm~\ref{alg:pivot} subdivides partition classes according to $N(p)$ as follows:

\begin{remark} \label{rem:pivot}
Suppose that $A,B$ are partition classes (in order) that replace a partition class $P$ during algorithm~\ref{alg:pivot}.  Then $A,B \ne \emptyset$, $p \notin P$, and if $p$ is in an earlier partition class than $P$, then $A \cap N(p) = \emptyset$ and $B \subseteq N(p)$, otherwise $A \subseteq N(p)$ and $B \cap N(p) = \emptyset$.  Moreover, if $P$ is not replaced during algorithm~\ref{alg:pivot}, then either $p \in P$, $P \cap N(p) = \emptyset$, or $P \subseteq N(p)$.
\end{remark}

\begin{algorithm} [ht]

\BlankLine

\KwIn{An ordered partition $\mathcal{P} = P_1,\ldots,P_k$ and a distinguished vertex $p \in P_i$.}

\KwOut{A refinement of $\mathcal{P}$.}

\BlankLine

\ForEach{$P_j \in \mathcal{P}, j \ne i$ such that $N(p)$ splits $P_j$} {

\BlankLine

$A \gets P_j - N(p)$\;
$B \gets P_j - A$\;

\BlankLine

\lIf{$i < j$}{replace $P_j$ in $\mathcal{P}$ with $A,B$ in order}
\lElse{replace $P_j$ in $\mathcal{P}$ with $B,A$ in order}}
\BlankLine
\KwRet{$\mathcal{P}$}\;

\BlankLine

\caption{$Pivot(\mathcal{P},p)$ \cite{m.mcconnell:linear-time}}  \label{alg:pivot}
\end{algorithm}

Observe that in algorithm~\ref{alg:pivot}, it may happen that one of $A$ or $B$ is the empty set.  By convention, throughout this paper, we take the approach that such an empty set is not inserted into the partition as a partition class.  Later in this paper, we will have need to generalize the above pivot operation.  There too we adopt this convention.  In general, whenever partitions and refinement are concerned in this paper, the empty set is not permitted.

The pivot operation above can be used to compute a source vertex in a prime comparability graph according to algorithm~\ref{alg:source}~\cite{m.mcconnell:linear-time}.  If the partition class $P$ at the end of algorithm~\ref{alg:source} did not contain a single vertex, then it would be a non-trivial module, by remark~\ref{rem:pivot}, contradicting the input graph being prime.  The fact that this vertex must be a source (or equivalently, a sink) vertex is due to the following easily verified invariant: throughout the algorithm, all edges between any vertex in $P$ and it neighbours outside $P$ must be directed the same way.  Applying algorithm~\ref{alg:source} to the graph in figure~\ref{fig:example}, where the sequence of pivots chosen during its execution is $a,z,b,u,y,q,w,r,v,x$, results in the vertex $x$ correctly being returned as a source vertex.

\begin{algorithm} [ht]

\BlankLine

\KwIn{A prime comparability graph $G$.}

\KwOut{A source vertex in the transitive orientation of $G$.}

\BlankLine

pick some vertex $x \in V(G)$\;

\BlankLine
let $\mathcal{P}$ be the ordered partition $\{x\},V(G) - \{x\}$\;
let $P$ be $\mathcal{P}$'s last partition class (throughout what follows)\;

\BlankLine

\lWhile{there is a vertex $z \notin P$ that has not been pivot} { $Pivot(\mathcal{P},z)$}

\BlankLine

\KwRet{$P$}\;

\BlankLine

\caption{$Source(G)$~\cite{m.mcconnell:linear-time}} \label{alg:source}
\end{algorithm}

An ordered partition $\mathcal{P}$ of $V(G)$ is \emph{consistent with a linear extension} if there is a way of permuting the vertices in each of its partition classes that results in a linear extension.  Observe that if $x$ is a source vertex, then $\{x\},V(G) - \{x\}$ is consistent with a linear extension.  Also observe that if the input to algorithm~\ref{alg:pivot} is consistent with a linear extension, then so too will be its output.  Finally, observe that any non-trivial partition of a prime graph contains a partition class $P$ that is split by some vertex $y \notin P$.  

The above observations suggest the following simple algorithm for computing a linear extension of a prime comparability graph $G$: first compute a source vertex $x$ using algorithm~\ref{alg:source}, then repeatedly apply algorithm~\ref{alg:pivot}, starting with the ordered partition $\{x\},V(G) - \{x\}$, until all partition classes are singleton sets.  This approach was first articulated in~\cite{m.mcconnell:linear-time} and later generalized to the non-prime case in~\cite{habib:stacs}.  The latter uses the following data-structure to represent ordered partitions:

\begin{itemize}

\item The ordered list of partition classes is represented as a doubly-linked list.

\item The elements within each partition class are represented as a doubly-linked list.

\item Each partition class maintains two indices indicating its \emph{range}: the number of elements in partition classes before it, plus the number of its own elements.

\item Each partition class maintains a pointer to the first element in the doubly-linked list representing its elements.

\item Each element within a partition class maintains a \emph{parent} pointer to its containing partition class.

\end{itemize}

\noindent
The corresponding implementation of algorithm~\ref{alg:pivot} proceeds as follows:

\begin{itemize}

\item Traverse the adjacency list of the pivot to determine the partition classes different than its own that contain at least one of its neighbours.

\item Use the range indices to determine the relative position of the partition class containing the pivot and those encountered above.

\item For each partition class encountered above, create an empty partition class immediately before those that appear earlier than the pivot, and a new partition class immediately after those that appear after the pivot.

\item Traverse the adjacency list of the pivot once more, and move each neighbour to the adjacent empty partition class created above, updating the range of the affected partition classes to reflect this change.
\end{itemize}

Algorithm~\ref{alg:pivot} clearly runs in time $O(|N(p)|)$ using the above implementation.  The efficiency of the approach described above for computing linear extensions therefore depends on the number of times algorithm~\ref{alg:pivot} is invoked.  To limit the number of times a vertex is a pivot,~\cite{m.mcconnell:linear-time} developed an ingenious rule: only reuse a vertex as pivot when its containing partition class is at most half the size it was the last time it was pivot.  The result is an $O(m\log{n})$ transitive orientation algorithm for prime graphs.  The same rule was applied in~\cite{habib:stacs} for the non-prime case, the result also being an $O(m\log{n})$ algorithm.

\subsection{Lexicographic Breadth-First Search} \label{sec:LBFS}

In Breadth-First Search, the order in which a vertex is explored depends only on its earliest explored neighbour.  Lexicographic Breadth-First Search (LBFS) goes further by using all previously explored neighbours to determine when a vertex is explored.  This is accomplished through the lexicographic labelling scheme specified in algorithm~\ref{alg:LBFS}: throughout the algorithm vertices carry lexicographic labels that have been assigned to them by their previously explored neighbours -- vertices with earlier explored neighbours will have lexicographically lager labels -- and on each iteration of the algorithm, a vertex with lexicographically largest label is chosen.  It is important to note that on any iteration of the algorithm, many vertices may share the lexicographically largest label.

To understand the operation of algorithm~\ref{alg:LBFS}, consider its operation on the graph in figure~\ref{fig:example}.  One possible output of the algorithm is $x,b,y,u,z,q,w,r,v,a$.  To see why, notice that initially, all vertices share the same label, and therefore any one of them can be chosen, say $x$.  From there, $x$ labels its only neighbour, $b$, and all other vertices retain their empty label.  Therefore $b$ has lexicographically largest label and must be chosen.  From there, $b$ labels its neighbours, while others retain their empty label, and so on.

\begin{algorithm} 

\BlankLine

\KwIn{A graph $G$.}

\KwOut{An ordering $\sigma$ of $V(G)$.}

\BlankLine

initialize $\sigma$ as the empty list\;
initialize each vertex in $V(G)$ by assigning it an empty label\;

\BlankLine

\While{$|\sigma| \ne |V(G)|$} {
\BlankLine
select a vertex $x \notin \sigma$ having lexicographically largest label and append it to $\sigma$\;

\lForEach{vertex $y \in N(x) - \sigma$} {append $|V(G)| - |\sigma| + 1$ to $y$'s label}
\BlankLine

}

\BlankLine

\KwRet{$\sigma$}\;

\BlankLine

\caption{$LBFS(G,x)$~\cite{rose:lexbfs}} \label{alg:LBFS}
\end{algorithm}

As we saw above, the initial vertex of any lexicographic breadth-first search is selected arbitrarily.  A variant of algorithm~\ref{alg:LBFS} allows this initial vertex to be specified as an input to the algorithm.  This is accomplished by assigning the input vertex some large initial label, say $\infty$ prior to any iteration of the while loop.  In proving results and defining things about LBFS we will use algorithm~\ref{alg:LBFS} as it is specified above.  However, it will be convenient throughout this paper to sometimes refer to the variant taking an input vertex, with the intention being clear from context.

Although algorithm~\ref{alg:LBFS} produces an ordering of vertices, more important for this paper will be the labels it assigns in doing so.  These are important in so far as they define \emph{slices}: any set of vertices sharing the lexicographically largest label immediately before an iteration of the while-loop in algorithm~\ref{alg:LBFS}.  Hence, there are $|V(G)| = n$ different slices, one for every iteration of the while-loop, and the set of slices can be ordered according to the iterations on which the slices are defined.  The vertex selected from the slice and added to $\sigma$ on each iteration is the slice's \emph{initial vertex}.  Notice that the set $V(G)$ is always the first slice.  Building on our example earlier, one possible ordered set of slices and initial vertices for algorithm~\ref{alg:LBFS} on the input graph in figure~\ref{fig:example} is the following:

\begin{itemize}

\item $S_1 = \{x,b,y,u,v,r,q,z,w\}$, initial vertex $x$;

\item $S_2 = \{b\}$, initial vertex $b$;

\item $S_3 = \{y,u,v,r,q,z,w\}$, initial vertex $y$;

\item $S_4 = \{u,v,r,q,z,w\}$, initial vertex $u$;

\item $S_5 = \{z,q,w\}$, initial vertex $z$;

\item $S_6 = \{q,w\}$, initial vertex $q$;

\item $S_7 = \{w\}$, initial vertex $w$;

\item $S_8 = \{r\}$, initial vertex $r$;

\item $S_9 = \{v\}$, initial vertex $v$;

\item $S_{10} = \{a\}$, initial vertex $a$.

\end{itemize}

It is not difficult to see the hierarchical relationship between slices in the example above.  This is formalized in the following remark:

\begin{remark} \label{rem:sliceContainment}
If $S$ and $S'$ are two distinct slices, then either $S \subset S'$, $S' \subset S$, or $S \cap S' = \emptyset$.   
\end{remark}

This leads to the idea of \emph{maximal subslices}: $S'$ is a maximal subslice of $S$ if $S' \subset S$ and there is no other slice $S''$ such that $S' \subset S'' \subset S$.  By convention, the initial vertex of a slice is also considered one of its maximal subslices.  Further, throughout this paper, we will generally use $S_1,\ldots,S_n$ to denote the set of all slices, and $x,S_1,\ldots,S_k$ to denote a set of maximal subslices, where $x$ is the initial vertex.  

Observe that the set of maximal subslices is clearly a partition of that slice.  And of course, the ordering of all slices defines an ordering of any set of maximal subslices (with the convention that the initial vertex always appears first).  

An edge is \emph{active} for a slice $S$ if its endpoints reside in different maximal sublices of $S$.  The notion of active edges is adapted from~\cite{dahlhaus:modular}, where it was introduced for the hierarchical set of strong modules that arise in modular decomposition.  Their reason for doing so relates to the following observation, presented here in terms of slices.
  
\begin{remark} \label{rem:activeOnce}
For any edge there is a single slice for which it is active.
\end{remark}

The following are properties of active edges specific to the LBFS context, each a direct consequence of the labelling scheme employed by algorithm~\ref{alg:LBFS}:

\begin{remark} \label{rem:maximalSlice}
Let $S$ be a slice and $x,S_1,\ldots,S_k$ its maximal subslices in order.  Then:

\begin{enumerate}

\item Either $N(x) \cap S = S_1$ or $N(x) \cap S = \emptyset$;

\item Every vertex $y \in S_i$ is either universal to or isolated from $S_j, i < j$;

\item For every $S_i$ and $S_j, 1< i < j$, there is a vertex $y \in S_{\ell}, \ell < i$, such that $y$ is universal to $S_i$ and isolated from $S_j$.

\end{enumerate}

\end{remark}

The last two remarks form the basis for the transitive orientation algorithm presented in this paper.  They imply many other properties of active edges that will be presented as needed throughout the paper.  One such property is the following:

\begin{lemma} \label{lem:sameDirection}
Consider an LBFS ordering of a comparability graph $G$.  Let $S$ be one of the slices and ${x},S_1,\ldots,S_k$ its maximal sub slices in order.  Assume that $N(x) \cap S \ne \emptyset$.  Then in any transitive orientation, the edges between all vertices in some $S_i, i > 1$, and any vertex $z \notin S_i$ are all directed the same way.
\end{lemma}

\begin{proof}
Assume that $z \in S_j$.  If $j < i$, then the result follows by condition~3 of remark~\ref{rem:maximalSlice}  So assume that $j > i$.  Then by the same remark, there is a vertex $q \in S_{\ell}, \ell < i$ that is universal to $S_i$ but isolated from $z$.  As with $z$, all edges from $q$ to $S_i$ must be directed the same way.  Therefore all edges between $S_i$ and $z$ must also be directed the same way.
\end{proof}

\subsection{Overview} \label{sec:overview}

This paper develops an algorithm to compute a linear extension of a prime comparability graph.  It uses the same basic approach as was described earlier in the context of partition refinement.  The difference is that instead of limiting the number of times a vertex is pivot, this paper limits what edges incident to a pivot are used each time.  The set of edges incident to each vertex will be partitioned, and every time a vertex is pivot, a unique subset of these incident edges will be processed.  In the end, each edge will only be processed once, allowing for linear-time transitive orientation, $O(n + m)$.

Partition refinement will remain the primary mechanism for computing the linear extension.  Algorithm~\ref{alg:pivot} will need to be generalized so that only a subset of vertices in a pivot's neighbourhood is processed.  The partition of the set of edges incident to each vertex will be defined in terms of the slices for which each edge is active.  This will necessitate computing an LBFS ordering, the corresponding set of slices, and the active edges they define.  The co-components of the graph induced by each slice will also need to be computed and treated as a kind of pseudo-slice for reasons that will be clear later.

The starting point for the algorithm is the same as before: computing a source vertex $x$ using algorithm~\ref{alg:source}.  The necessary LBFS ordering will be initiated from this source vertex.  Just as before, $\{x\},V(G) - \{x\}$ is consistent with a linear extension and will be used as the initial ordered partition.  Each slice will then be processed (in order) and two rounds of refinement undertaken for each.  Both will ensure that the resulting ordered partition remains consistent with a linear extension.  

The first round of refinement targets the initial vertex of the slice in question, ensuring that by the end, none of its neighbours share its partition class.  This has the effect of determining the directions on all the edges incident to that initial vertex.  With that, the directions on all other edges active for that slice can be determined using a second round of refinement.  To facilitate the first round of refinement, algorithm~\ref{alg:pivot} will be adapted so that only a single partition class (the one containing the initial vertex) is subdivided.  

It is critical for the running-time that during both rounds of refinement, only the edges active for that slice are processed.  However, as will be described later, knowing the co-components for each slice will allow the algorithm to ``cheat'' a little in this regard.

The rest of the paper is organized as follows.  Section~\ref{sec:initialization} outlines the various initialization steps that are needed, along with their correctness and running-time.  Next, section~\ref{sec:refinement} describes how refinement can produce a linear extension, and also includes correctness and running-time.  The paper concludes in section~\ref{sec:conclusion}.

\section{Initialization} \label{sec:initialization}

The following is needed before refinement can be applied to compute the desired linear extension: 

\begin{itemize}

\item A source vertex;

\item An LBFS ordering initialized from that source vertex;

\item The corresponding slices and active edges; and 

\item The co-components of each slice.  

\end{itemize}

Note that a source vertex can be computed in $O(n+m)$ time using algorithm~\ref{alg:source} based on the partition refinement implementation of algorithm~\ref{alg:pivot} described in section~\ref{sec:partitionRefinement}; and using that same implementation, algorithm~\ref{alg:LBFS} is known to run in $O(n+m)$ time as well.~\cite{habib:stacs}  The remaining initialization components are addressed below in turn.

\subsection{Slices} \label{sec:slices}

Explicitly computing the set of vertices in each slice would prove too costly.  Instead, they will be implicitly computed by constructing the tree defined by the containment relationship between slices (see remark~\ref{rem:sliceContainment}) as defined below:

\begin{definition}
Let $\mathcal{S}$ be the set of slices defined by some LBFS ordering of the graph $G$.  The corresponding \emph{slice-tree} is defined such that its interior nodes correspond to the elements of $\mathcal{S}$, the leaves correspond to the vertices in $G$, and the children of each slice are its maximal subslices.
\end{definition}

An example slice-tree is provided in figure~\ref{fig:sliceTree}.  The following are all easily verified properties of the slice-tree, each an extension of the containment relationship captured by remark~\ref{rem:sliceContainment}.

\begin{figure}  
\begin{center}  
\includegraphics[scale=0.85]{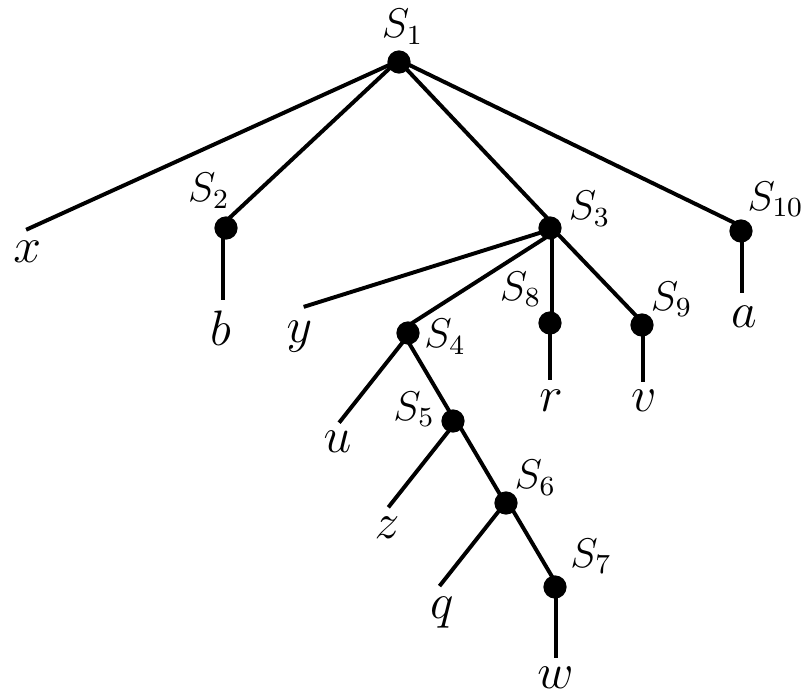}
\caption{A slice-tree for the graph in figure~\ref{fig:example}, in this case corresponding to the example execution of algorithm~\ref{alg:LBFS} described in section~\ref{sec:LBFS}.}
\label{fig:sliceTree}
\end{center}   
\end{figure}

\begin{remark} \label{rem:sliceTreeAncestor}
Let $\mathcal{S} = S_1,\ldots,S_n$ be the set of slices defined by some LBFS ordering and consider the corresponding slice-tree. 

\begin{enumerate}

\item If $S_i$ is an ancestor of $S_j$, then $i < j$;  

\item $S_i$ is an ancestor of $S_j$ if and only if $S_j \subset S_i$;

\item $y \in S_i$ if and only if $S_i$ is an ancestor of $y$.

\end{enumerate}

\end{remark}

Recall the labels maintained by vertices during algorithm~\ref{alg:LBFS}.  Notice that after each vertex is appended to the output ordering $\sigma$, its label remains fixed -- or final.  When discussing vertices in the context of the slice-tree, this final label is what is meant when referring to a vertex's label.  
Of course, slices are defined by labels, and these will be what is meant when referring to a slice's label in the slice-tree context.  In particular, the first slice consisting of all vertices is assigned the empty label.  Notice the correspondence between the labels defining slices and the labels assigned to their initial vertex.  The following are properties of the slice-tree in terms of these labels:

\begin{lemma} \label{lem:labels}
Consider the LBFS ordering $\sigma$ and its corresponding slice-tree.

\begin{enumerate}

\item If the slice $S$ is an ancestor of the vertex $x$, then $S$'s label is a (not necessarily proper) prefix of $x$'s label.

\item If $x$ and $y$ are consecutive vertices in $\sigma$, then the label of the slice being their least common ancestor is the longest (not necessarily proper) prefix shared by $x$'s label and $y$'s label.

\end{enumerate}

\end{lemma}

\begin{proof}
For item~1, notice that $x \in S$ by item~3 of remark~\ref{rem:sliceTreeAncestor}.  Thus, $x$ and $S$ shared the same label at the time that $S$ became a slice.  It follows that $S$'s label is a (not necessarily proper) prefix of $x$'s label.  For item~2, let $S$ be the least common ancestor of $x$ and $y$.  Then $S$'s label is a (not necessarily proper) prefix of $x$'s label and $y$'s label by item~1; it is the longest such prefix by definition of $S$ being their least common ancestor.
\end{proof}

Knowing lemma~\ref{lem:labels}, and remembering the correspondence between vertices and slices, it is now easy to compute the slice-tree as in algorithm~\ref{alg:sliceTree}.  The correctness of algorithm~\ref{alg:sliceTree} follows immediately from the relevant definitions, lemma~\ref{lem:labels}, and the fact that vertices are processed in order.  For its implementation, note that identifying $z$ is clearly a constant time operation.  Obviously $\rho$ can be computed in time on the order of the length of $y$'s label.  And once $S$ is identified, the updates to $T$ can be performed in constant time.      

\begin{algorithm} 

\BlankLine

\KwIn{An LBFS ordering $\sigma$ with initial vertex $x$.}

\KwOut{The slice-tree corresponding to the set of slices defined by $\sigma$.}

\BlankLine

\BlankLine

initialize $T$ as the rooted tree with a single interior node having $x$ as its first (and only) child\;
initialize the root of $T$ by assigning it an empty label\;

\BlankLine

\ForEach{$y \in \sigma - \{x\}$ (in order)} {

\BlankLine

let $z$ be the vertex preceding $y$ in $\sigma$\;
let $\rho$ be the longest (not necessarily proper) prefix shared by $y$'s label and $z$'s label\;
let $S$ be the ancestor of $z$ whose label is $\rho$\;

\BlankLine

update $T$ by creating a new node $u$ and making it the last child of $z$\;
update $T$ by making $y$ the first and only child of $u$\;

\BlankLine

}

\BlankLine

\KwRet{$T$}\;

\BlankLine

\caption{$SliceTree(\sigma)$} \label{alg:sliceTree}
\end{algorithm}

Now consider the problem of identifying $S$.  This can be accomplished by traversing the path from $z$ to the root of the slice-tree, comparing $\rho$ with the label of each slice thus encountered, stopping when a match is found.  Of course, the last character in those two labels will be the same.  And once those last two characters match, all the others must match as well because both are (not necessarily proper) prefixes of $z$'s label ($\rho$ by definition and the other by item~1 of lemma~\ref{lem:labels}).  Thus, each comparison is a constant-time operation.

The problem for the efficiency of algorithm~\ref{alg:sliceTree} is that several consecutive slices on the path from $z$ to the root of the slice-tree might share the same label.  If every label on the path was distinct, then identifying $S$ would be on the order of the length of $z$'s label.  This would mean the total cost of algorithm~\ref{alg:sliceTree} would be $O(n+m)$.  The simple solution is to maintain a copy of the path from $z$ to the root, excluding those nodes whose parents share their labels.  Updating this path is clearly a constant time operation after each new vertex $y$ is inserted.  The total cost therefore remains $O(n+m)$.

\subsection{Active Edges} \label{sec:activeEdges}

Active edges will drive the refinement that will later be used to compute the desired linear extension.  Given their importance, special notation is developed for them below:

\begin{notation}
Let $S_1,\ldots,S_n$ be the set of slices in order.  For each vertex $x \in S_i$, let $\alpha_i(x)$ be the set of vertices $y \in N(x)$ such that $xy$ is active for $S_i$.  If $x \notin S_i$, then $\alpha_i(x) = \emptyset$, by convention.
\end{notation}  

To illustrate the notion of active edges, recall the set of slices $S_1,\ldots,S_{10}$ defined earlier for the graph illustrated in figure~\ref{fig:example}.  The set of active edges for each of these slices is as follows:

\begin{itemize}

\item $S_1$: $xb, by, bu, bz, bq, bw, br, bv, za$;

\item $S_2$: $\emptyset$;

\item $S_3$: $yu,yz,yq,yw,ur,qv$;

\item $S_4$: $uz,uq,uw$;

\item $S_5$: $\emptyset$;

\item $S_6$: $qw$;

\item $S_7$: $\emptyset$;

\item $S_8$: $\emptyset$;

\item $S_9$: $\emptyset$;

\item $S_{10}$: $\emptyset$.

\end{itemize}

The slice-tree can be used to efficiently compute active edges in the manner outlined by algorithm~\ref{alg:activeEdges} below.  Its correctness is based on lemma~\ref{lem:activeEdges}.  Notice that algorithm~\ref{alg:activeEdges} assigns labels to slices.  These are not to be interpreted as the labels assigned during algorithm~\ref{alg:LBFS}.  Furthermore, the labels referenced in lemma~\ref{lem:activeEdges} and corollary~\ref{cor:activeEdges} afterward refer to those defined by algorithm~\ref{alg:activeEdges}, not algorithm~\ref{alg:LBFS} which was the custom in the previous section.

\begin{algorithm} 

\BlankLine

\KwIn{A slice-tree $T$ corresponding to the set of slices $\mathcal{S} = S_1,\ldots,S_n$ defined by some LBFS ordering.}

\KwOut{The slice-tree $T$ augmented so that every slice has associated with it a graph induced by that slice's active edges.}

\BlankLine
\BlankLine

initialize every slice $S_i$ by assigning it an empty label\;
initialize every slice $S_i$ by assigning it an empty graph referenced by $G(S_i)$\;

\BlankLine

\ForEach{$S_i$ (in order)} {

\BlankLine

let $y$ be the initial vertex for $S_i$\;

\lForEach{$S_j, j > i,$ that is maximal with respect to $S_j \subseteq N(y)$} { add $y$ to $S_j$'s label}

\BlankLine

}

\BlankLine

\ForEach{$S_j$ with non-empty label} {

\BlankLine
let $S_i$ be the parent of $S_j$\;

\ForEach{pair $y,z$ where $y \in S_j$ and $z$ appears in $S_j$'s label} {

update $G(S_i)$ with the edge $yz$\;}

}

\BlankLine

\KwRet{$T$}\;

\BlankLine

\caption{$ActiveEdges(T)$} \label{alg:activeEdges}
\end{algorithm}

\begin{lemma} \label{lem:activeEdges}
Let $S$ be a slice and $x,S_1,\ldots,S_k$ its maximal subslices in order.  If $w$ is in the label for $S_i$, then $w \in S - S_i$ and $S_i \subseteq N(w)$.  Conversely, for any vertex $y \in S_i$: 

\begin{enumerate}

\item If $x \in N(y)$, then $x$ appears in the label of $S_i$; 

\item If there is a $z \in N(y) \cap S_j, j < i$, then $z$ appears in the label of $S_i$.

\end{enumerate}

\end{lemma}

\begin{proof}
It is clear from algorithm~\ref{alg:activeEdges} that $S_i \subseteq N(w)$.  It follows that $w \notin S_i$.  Notice as well from algorithm~\ref{alg:activeEdges} that there is no ancestor $S'$ of $S_i$ such that $S' \subseteq N(w)$.  Therefore, by either item~1 or item~2 of remark~\ref{rem:maximalSlice}, it must happen that $w \in S$.  

Now consider some $y \in S_i$ and an ancestor $S'$ of $S_i$.  Therefore $x \in S'$, and so $S' \not \subseteq N(x)$.  Hence, $S'$ is not a maximal slice such that $S' \subseteq N(x)$.  If $x \in N(y)$, then $N(x) \cap S \ne \emptyset$, and therefore $N(x) \cap S = S_1$, by item 1 of remark~\ref{rem:maximalSlice}.  So in that case, $i = 1$ and $S_i \subseteq N(x)$.  It follows that $x$ will be added to $S_i$'s label.  

Finally, consider some $z \in N(y) \cap S_j, j < i$.  As with $x$, no ancestor $S'$ of $S_i$ can be a maximal slice such that $S' \subseteq N(z)$, since $z \in S'$.  But $S_i \subseteq N(z)$ by item 2 of remark~\ref{rem:maximalSlice}.  It follows that $z$ will be added to $S_i$'s label.   
\end{proof}

\begin{lemma}
Algorithm~\ref{alg:activeEdges} is correct.  
\end{lemma}

\begin{proof}
Let $S$ be a slice and $x,S_1,\ldots,S_k$ its maximal subslices in order.  Consider an active edge $xw$.  Then by item~1 of remark~\ref{rem:maximalSlice}, $w \in S_1$.  Furthermore, by item~1 of lemma~\ref{lem:activeEdges}, $x$ appears in the label of $S_1$.  Clearly the edge $xw$ is added to $G(S)$ in this case.  The same obviously holds for the active edge $yz$ where $y \in S_i$ and $z \in S_j, i < j$, although this time by item~2 of lemma~\ref{lem:activeEdges}.  

Now consider an edge $uv$ in $G(S)$.  According to algorithm~\ref{alg:activeEdges}, it can be assumed without loss of generality that $u \in S_i$ and $v$ appears in $S_i$'s label.  Then by lemma~\ref{lem:activeEdges}, $v \in S - S_i$ and $S_i \subseteq N(v)$.  It follows that $uv$ is an active edge for $S$.
\end{proof}

Consider the slices with empty labels following algorithm~\ref{alg:activeEdges}.  Let $S$ be a slice, $x,S_1,\ldots,S_k$ its maximal subslices in order, and assume that $S_i$ has an empty label.  Then by lemma~\ref{lem:activeEdges}, the initial vertex $x$ is isolated from $S_i$ as are all other vertices $y \in S_j, j < i$.  But what about vertices $z \in S_{\ell}, \ell > i$?  In fact, no such vertex can exist, by item~3 of remark~\ref{rem:maximalSlice}.  This leads to the following corollary, which characterizes the relationship between active edges and the labels assigned by algorithm~\ref{alg:activeEdges}.

\begin{corollary} \label{cor:activeEdges}
Consider a slice $S_i$ and one of its maximal subslices $S_j$.  If $w$ is in the label for $S_j$, then for every vertex $y \in S_j$, $w \in \alpha_i(y)$.  Moreover, if $S_j$'s label is empty, then for every vertex $y \in S_j, \alpha_i(y) = \emptyset$.
\end{corollary}

\begin{proof}
Follows directly from lemma~\ref{lem:activeEdges} and the relevant definitions.
\end{proof}

In other words, the vertices of each slice without incident active edges are those in slices not having labels at the end of algorithm~\ref{alg:activeEdges}.  Such vertices and their containing slices are further characterized below.  Understanding them is crucial for establishing the running-time of the transitive orientation algorithm developed in this paper.

\begin{definition}
Let $S_1,\ldots,S_n$ be the set of slices in order.  A vertex $x \in S_i$ is said to be \emph{disconnected} for $S_i$ if $\alpha_i(x) = \emptyset$ and is \emph{connected} for $S_i$ otherwise.
\end{definition}

\begin{remark} \label{rem:disconnected}
Let $S$ be a slice and $x,S_1,\ldots,S_k$ its maximal subslices in order.  Then $G[S]$ is disconnected if and only if there exists a disconnected vertex for $S$.  Moreover, when there is a vertex $y$ that is disconnected for $S$, it happens that $y \in S_k$, and in this case, all vertices in $S_k$ are disconnected for $S$.  Furthermore, when $S$ is disconnected, one of the components of $G[S]$ is $\{x\} \cup S_1 \cup \cdots \cup S_{k-1}$, the others being those of $G[S_k]$.
\end{remark}

Remark~\ref{rem:disconnected} will be used below to guarantee that no disconnected vertices are processed during algorithm~\ref{alg:activeEdges}.  This will be important to guarantee the linear running time of algorithm~\ref{alg:activeEdges} above.   

There are two blocks of nested for-loops in algorithm~\ref{alg:activeEdges}.  The case with disconnected vertices occurs for the second one, which will be considered below.  For the first block of nested for-loops, notice that a simple bottom-up marking scheme can compute the necessary $S_j$'s for each vertex $y$: start by marking all leaves in $N(y)$, and then recursively mark each node whose children are all marked.  The time this takes is on the order of the size of the subtree rooted at each $S_j$.  

Now, every slice contains at least one vertex: its initial vertex.  Hence, if a slice contains only one vertex, then all of its ancestors must contain at least two vertices.  Thus, the size of the subtree rooted at each $S_j$ is on the order of the size of the number of its descendent leaves, or equivalently, the size of $S_j$.  But the $S_j$'s are all disjoint, by definition, and each $S_j \subseteq N(y)$.  It follows that the running time of the first block of nested for-loops is $O(n+m)$.

Having finished the first block of nested for-loops, it can be assumed that the labels on each slice are in place.  To implement the second block, it will be necessary to compute the set of vertices in each slice whose label is non-empty.  This can be done recursively as follows.  Let $S$ be a slice and $x,S_1,\ldots,S_k$ its maximal subslices in order.  Assume that the set of vertices in each $S_i$ has been computed and is represented as a linked-list.  For each $S_i$ whose label is non-empty, make a copy of the list.

For the recursion to work, a list of all vertices in $S$ must now be computed.  Observe that remark~\ref{rem:disconnected} together with corollary~\ref{cor:activeEdges} implies that there is at most one $S_i$ whose label is empty, namely $S_k$.  If there is such an $S_i$, take its recursively computed list of vertices and append to it the just computed copies of the lists from all other maximal subslices (including a copy of $x$).  This produces the desired list of vertices for $S$.  Every vertex in a slice with a label has at least one incident active edge, by corollary~\ref{cor:activeEdges}.  Hence, the cost of creating the list of vertices for $S$ is proportional to the total number of active edges for $S$.  But by remark~\ref{rem:activeOnce}, each edge is active for exactly one slice.  Therefore the total cost of creating these lists for all slices is $O(n+m)$.  

Finally, observe that once the set of labels has been computed for each slice, and the set of vertices has been computed for each slice with a non-empty label, the graph $G(S)$ can clearly be computed in time on the order of the number of active edges for $S$.  As above, the cost of doing so for all slices is $O(n+m)$.

Before moving on, it will be necessary to introduce the following results concerning active edges.  They will be needed later to prove the correctness of the refinement that produces the desired linear extension.  Each either rephrases or extends remark~\ref{rem:maximalSlice}.

\begin{remark} \label{rem:neighbour}
Let $S_1,\ldots,S_n$ be the set of slices in order.  Consider some slice $S_i$ with initial vertex $x$.  If $N(x) \cap S_i \ne \emptyset$, then $N(x) \cap S_i = S_{i+1} = \alpha_i(x)$.
\end{remark}  

\begin{lemma} \label{lem:universalSlice}
Let $S$ be a slice and $x,S_1,\ldots,S_k$ its maximal subslices in order.  Assume that $N(x) \cap S \ne \emptyset$.  Then for every $S_i, i > 1$, that contains a connected vertex for $S$, there exists a $y \in S_{\ell}, \ell < i$, that is universal to $S_i$. 
\end{lemma}

\begin{proof}
If $i \ne k$, then the result follows by item~3 of remark~\ref{rem:maximalSlice}.  So assume that $i = k$.  Suppose for contradiction that no such vertex $y$ exists.  Then there is no vertex in $S - S_k$ that is adjacent to any vertex in $S_k$, by item~2 of remark~\ref{rem:maximalSlice}.  Therefore $S_i = S_k$ would contain a disconnected vertex, which contradicts remark~\ref{rem:disconnected}.
\end{proof}

\begin{lemma} \label{lem:activeEdgePartition}
Let $\mathcal{S} = S_1,\ldots,S_n$ be the set of slices defined by some LBFS ordering.  If $y$ is the initial vertex for $S_i$, then $\alpha_j(y) = \emptyset$ for all $j > i$. 
\end{lemma}

\begin{proof}
Since $y$ is the initial vertex of $S_i$, if $y \in S_j$, then either $i = j$ or $S_j$ is an ancestor of $S_i$.  Therefore, $j \le i$, by item~1 of remark~\ref{rem:sliceTreeAncestor}.  So by definition, $\alpha_j(y) = \emptyset$ for all $j > i$.  
\end{proof}

\subsection{Co-Components} \label{sec:coComponents}

Section~\ref{sec:overview} hinted at the need to compute the co-components of the graph induced by each slice.  To see why it is useful, consider a slice $S$ and its maximal subslices $x,S_1,\ldots,S_k$ in order.  Recall that if $N(x) \cap S \ne \emptyset$, then $N(x) \cap S = S_1$, by item~1 of remark~\ref{rem:maximalSlice}.  Let $C$ be a co-component of $G[S_1]$.  Knowing the direction of an active edge incident to any vertex in $C$ is equivalent to knowing the direction of the active edges incident to all vertices in $C$ as follows:

\begin{lemma} \label{lem:coSameDir}
Consider an LBFS ordering of a comparability graph $G$.  Let $S$ be one of the slices and $x,S_1,\ldots,S_k$ its maximal subslices in order.  Assume that $N(x) \cap S \ne \emptyset$.  Let $C$ be a co-component of $G[S_1]$.  Then in any transitive orientation, all edges between $C$ and the set $S - C$ are directed the same way.
\end{lemma}  

\begin{proof}
Consider some transitive orientation of $G$.  Recall that $x$ must be universal to $C$ since $N(x) \cap S = S_1$, by item~1 of remark~\ref{rem:maximalSlice}.  Partition $C$ into $C_{out}$ and $C_{in}$, the former being those vertices in $C$ whose shared edge with $x$ is directed toward $x$, the latter defined as those vertices in $C$ whose shared edge with $x$ is directed away from $x$.  Suppose for contradiction that $C_{in} \ne \emptyset$ and $C_{out} \ne \emptyset$.  Then by virtue of $G$ being a comparability graph, there must be a join between $C_{in}$ and $C_{out}$, contradicting $C$ being a co-component of $G[S_1]$.  Since $x$ is not adjacent to any vertex in $V(G) - N[x]$, it follows that all edges between $C$ and $S_2 \cup \cdots \cup S_k$ must all be directed the same way as the edges between $C$ and $x$.
\end{proof}

In particular, the above lemma says that once the direction on a single edge between $x$ and $C$ is known, then all the directions of the edges between $x$ and $C$ are known (and the directions are all the same).  This will be useful during the first round of refinement described in section~\ref{sec:overview} where the goal is to determine the directions on the active edges incident to $x$.  During that round of refinement, the vertices in $C$ can be moved as one block.  This round of refinement therefore does not technically rely only on active edges.  This was the ``cheat'' alluded to in section~\ref{sec:overview}.

It is a simple matter to compute the co-components of each slice once the active edges for each vertex are known.  The following result characterizes these co-components and translates directly into algorithm~\ref{alg:coComponents}, which computes the co-components of any single slice.  

\begin{lemma} \label{lem:coComponents}
Let $S$ be a slice and $x,S_1,\ldots,S_k$ its maximal subslices in order.  Let $C_1,\ldots,C_{\ell}$ be the co-components of $G[S_1]$.  Assume that $\overline{G[S]}$ is disconnected.  Then:

\begin{enumerate}

\item $N(x) \cap S = S_1$;

\item $x$ and $S_2,\ldots,S_k$ all belong to the same co-component of $G[S]$, call it $C$;

\item $C_i \subset C$ or $C_i \cap C = \emptyset$, for all $i$.

\end{enumerate}

\end{lemma}

\begin{proof}
If $N(x) \cap S \ne S_1$, then $N(x) \cap S = \emptyset$, by item 1 of remark~\ref{rem:maximalSlice}.  In that case, $x$ is a universal vertex in $\overline{G[S]}$, contradicting the latter being disconnected.  Moreover, $x$ must be universal to $S_2 \cup \cdots S_k$ in $\overline{G[S]}$, meaning they must be in the same co-component.  Now suppose for contradiction that $C_i$ overlaps $C$.  Say $A = C \cap C_i$ and $B = C_i - A$.  Then there is a join between $A$ and $B$, contradicting $C_i$ being a co-component of $G[S_1]$.  Therefore, either $C_i \subset C$ or $C_i \cap C = \emptyset$.
\end{proof}

\begin{algorithm} 

\BlankLine

\KwIn{A slice $S$.}

\KwOut{The co-components of $G[S]$.}

\BlankLine

let $x,S_1,\ldots,S_k$ be the maximal subslices of $S$\;

\BlankLine

\lIf{$N(x) \cap S \ne S_1$} { \KwRet{$S$}}

\BlankLine

let $C_1,\ldots,C_{\ell}$ be the co-components of $G[S_1]$\;

\BlankLine

initialize $\mathcal{C}$ as the empty set\;
initialize $C$ as $\{x\} \cup S_2 \cup \cdots S_k$\;

\BlankLine

\ForEach{$C_i$} { 

\lIf{$C_i$ is universal to $C$}{ $\mathcal{C} \gets \mathcal{C} \cup \{C_i\}$}
\lElse{$C \gets C \cup C_i$}
}

\BlankLine

$\mathcal{C} \gets \mathcal{C} \cup \{C\}$\;

\BlankLine

\KwRet{$\mathcal{C}$}\;

\BlankLine

\caption{$CoComponents(S)$} \label{alg:coComponents}
\end{algorithm}

\begin{lemma} \label{lem:coComponentsCorrect}
Algorithm~\ref{alg:coComponents} is correct.
\end{lemma}

\begin{proof}
If $N(x) \cap S \ne S_1$, then by item~1 of remark~\ref{rem:maximalSlice}, $N(x) \cap S = \emptyset$.  Therefore $x$ is a universal vertex in $\overline{G[S]}$, and so there is a single co-component $S$.  In this case algorithm~\ref{alg:coComponents} correctly outputs $S$.  So assume that $N(x) \cap S = S_1$. Suppose that $\overline{G[S]}$ is connected.  Then there can be no $C_i$ that is universal to $C$.  Notice that algorithm~\ref{alg:coComponents} correctly outputs $S$ in this case.  The case where $\overline{G[S]}$ is disconnected follows directly from lemma~\ref{lem:coComponents} and the relevant definitions.
\end{proof}

All it takes now is repeated application of algorithm~\ref{alg:coComponents} to compute the co-components of each slice.  This can be done recursively as was done earlier for algorithm~\ref{alg:activeEdges}.  Let $S$ be a slice and $x,S_1,\ldots,S_k$ its maximal subslices in order.  Assume that the co-components for each $S_i$ have been recursively computed and are represented as partitions of that slice using the data-structure for partition refinement from section~\ref{sec:partitionRefinement}.  

It will be necessary to assume that for each maximal subslice, a list of its vertices has been recursively computed.  This was also required by the implementation of algorithm~\ref{alg:activeEdges} and can be handled the same way.  It will be assumed as well that algorithm~\ref{alg:activeEdges} precedes algorithm~\ref{alg:coComponents} and hence the active edges for $S$ are already known.  

Now consider algorithm~\ref{alg:coComponents} operating on $S$ in this recursive context.  The maximal subslices of $S$ are defined by the slice-tree.  By remark~\ref{rem:neighbour}, the first conditional amounts to testing if $x$ has any incident active edges for $S$, which is a constant time operation since the active edges for each vertex are assumed to have been computed already.

If the first conditional succeeds, then the list of vertices in $S$ is returned.  But this happens automatically as part of the maintenance required to preserve the assumption that the list of vertices in each maximal subslice has been recursively computed.  As explained above, this can be handled as it is in the implementation of algorithm~\ref{alg:activeEdges}.

So assume that the first conditional fails in algorithm~\ref{alg:coComponents}.  The co-components of $G[S_1]$ have been recursively computed, as assumed above.  To initialize $C$, recall that the list of vertices in each maximal subslice has been recursively computed.  Copy the lists for the maximal subslices $S_2,\ldots,S_{k-1}$, merge these copies into a single list, and then append the result to the list for $S_k$.  Doing so avoids ``touching'' disconnected vertices in the case that $G[S]$ is disconnected (see remark~\ref{rem:disconnected}).  This was the same approach used in the implementation of algorithm~\ref{alg:activeEdges} (see section~\ref{sec:activeEdges}).  For the same reasons as there, the total cost over all slices is $O(n+m)$.

Determining if $C_i$ is universal to $C$ merely requires traversing the list of active edges for $S$ that are incident to each vertex in $C_i$.  Adding $C_i$ to $\mathcal{C}$ is a constant-time operation, and so too is adding $C$ to $\mathcal{C}$ after the loop.  However, adding $C_i$ to $C$ within the loop is on the order of the size of $C_i$.  But notice that each vertex in $C_i$ has at least one incident active edge, namely $x$.  Hence, adding $C_i$ to $C$ is on the order of the number of active edges for $S$.  The cost of all this work is therefore on the order of the total number of active edges for $S$.  Hence, the total cost over all slices is $O(n+m)$, by remark~\ref{rem:activeOnce}.  This and the preceding discussion implies that the co-components of all slices can be computed in time~$O(n+m)$.

\section{Refinement} \label{sec:refinement}

Consider a prime comparability graph $G$.  Based on the initialization work in section~\ref{sec:initialization}, it is now possible to assume the following has been computed in linear-time: 

\begin{itemize}

\item A source vertex $x \in V(G)$;

\item an LBFS ordering with $x$ as its initial vertex; 

\item the graph induced by the active edges for each slice defined by this LBFS ordering; and

\item the set of co-components for each of these slices.  

\end{itemize}

Given the above, a linear extension for $G$ can be computed in linear-time using partition refinement as described in the following subsections.

\subsection{Algorithm}

Two changes to the type of refinement defined by algorithm~\ref{alg:pivot} will be needed to produce a linear extension in linear-time.  The first involves using only a subset of a vertex's neighbours, not the whole neighbourhood as is done in algorithm~\ref{alg:pivot}.  A vertex will be pivot once for every one of their containing slices.  The subset of neighbours processed each time will be those whose corresponding edge is active for that slice.  One of the reasons this works relates to co-components and lemma~\ref{lem:coSameDir}.  Another is the following similar lemma:

\begin{lemma} \label{lem:sliceSameDir}
Consider an LBFS ordering of a comparability graph $G$.  Let $S$ be one of the slices and $x,S_1,\ldots,S_k$ its maximal subslices in order.  If $N(x) \cap S \ne \emptyset$, then in any transitive orientation, all edges between $S_i$ and $S_j$, $1 < i < j$, must be directed the same way.  
\end{lemma}

\begin{proof}
Assume for contradiction that the claim does not hold for some pair $S_i$ and $S_j$, $1 < i < j$.  Let $S_i$ be the leftmost maximal slice for which there is a corresponding $S_j$.  By remark~\ref{rem:maximalSlice}, there is a vertex $y \in S_{\ell}$, $\ell < i$, that is universal to $S_i$ and isolated from $S_j$.  So by choice of $S_i$, if $\ell > 1$, then all of $y$'s incident edges with $S_i$ must be directed the same way.  The same is also true if $\ell = 1$, by lemma~\ref{lem:coSameDir}.  It follows that all edges between $S_i$ and $S_j$ must be directed the same way: namely, the opposite direction of those between $y$ and $S_i$.
\end{proof}

Algorithm~\ref{alg:generalizedPivot} below generalizes the pivot operation so that only a subset of a vertex's neighbours are processed.  Once again, the input vertex $p$ is referred to as a \emph{pivot}.  Remark~\ref{rem:generalizedPivot} is the analogue of remark~\ref{rem:pivot}.  The running-time of algorithm~\ref{alg:generalizedPivot} is clearly $O(|S|)$ based on the implementation from section~\ref{sec:partitionRefinement}.    

\begin{algorithm} 

\BlankLine

\KwIn{An ordered partition $\mathcal{P} = P_1,\ldots,P_k$ of vertices, a distinguished vertex $p \in P_i$, and a set of vertices $S \subseteq N(p)$.}

\KwOut{A refinement of $\mathcal{P}$.}

\BlankLine
\BlankLine

\ForEach{$P_j \in \mathcal{P}, i \ne j,$ such that $S$ splits $P_j$} {

\BlankLine

$A \gets P_j - S$\;
$B \gets P_j - A$\;

\BlankLine

\lIf{$i < j$}{replace $P_j$ in $\mathcal{P}$ with $A,B$ in order}
\lElse{replace $P_j$ in $\mathcal{P}$ with $A,B$ in order}

\BlankLine

}

\BlankLine

\KwRet{$\mathcal{P}$}\;

\BlankLine

\caption{$GeneralizedPivot(\mathcal{P},p,S)$} \label{alg:generalizedPivot}
\end{algorithm}

\begin{remark} \label{rem:generalizedPivot}
Suppose that $A,B$ are partition classes (in order) that replace a partition class $P$ during algorithm~\ref{alg:generalizedPivot}.  Then $A,B,S \ne \emptyset$, $p \notin P$, and if $p$ is in an earlier partition class than $P$, then $A \cap S = \emptyset$ and $B \subseteq S$, otherwise $A \subseteq S$ and $B \cap S = \emptyset$.  Moreover, if $P$ is not replaced during algorithm~\ref{alg:generalizedPivot}, then either $p \in P$, $P \cap S = \emptyset$, or $P \subseteq S$.
\end{remark}

The second of the required changes to algorithm~\ref{alg:pivot} involves targeting a specific partition class: instead of allowing all partition classes to be subdivided as in algorithm~\ref{alg:pivot}, only one partition class specified in the input can be subdivided.  However, the order of the two classes subdividing the original partition class must additionally be allowed to vary.  This leads to algorithms~\ref{alg:pivotPush} and~\ref{alg:pivotPull} below.  As before, the input vertex $p$ is referred to as a \emph{pivot} in both cases.  Remarks~\ref{rem:pivotPush} and~\ref{rem:pivotPull} are the analogue of remark~\ref{rem:pivot}.  The running-time of each algorithm is clearly $O(|S|)$ based on the implementation from section~\ref{sec:partitionRefinement}.  

\begin{algorithm} 

\BlankLine

\KwIn{An ordered partition $\mathcal{P} = P_1,\ldots,P_k$ of vertices, a distinguished vertex $p \in P_i$ and partition class $P = P_j, j \ne i$, and a set of vertices $S$.}

\KwOut{A refinement of $\mathcal{P}$.}

\BlankLine
\BlankLine

\If{$S$ splits $P_j$} {
$A \gets P_j - S$\;
$B \gets P_j - A$\;
}

\BlankLine

\lIf{$i < j$}{replace $P_j$ in $\mathcal{P}$ with $A,B$ in order}
\lElse{replace $P_j$ in $\mathcal{P}$ with $B,A$ in order}

\BlankLine

\KwRet{$\mathcal{P}$}\;

\BlankLine

\caption{$PivotPush(\mathcal{P},p,P,S)$} \label{alg:pivotPush}
\end{algorithm}

\begin{remark} \label{rem:pivotPush}
Suppose that $A,B$ are partition classes (in order) that replace partition class $P_j$ during algorithm~\ref{alg:pivotPush}.  Then either $A \subseteq S$ and $S \cap B = \emptyset$ (if $j > i$ in the algorithm) or $B \subseteq S$ and $A \cap S = \emptyset$ (if $i < j$ in the algorithm).  
\end{remark}

\begin{algorithm} 

\BlankLine

\KwIn{An ordered partition $\mathcal{P} = P_1,\ldots,P_k$ of vertices, a distinguished vertex $p \in P_i$ and partition class $P = P_j, j \ne i$, and a set of vertices $S$.}

\KwOut{A refinement of $\mathcal{P}$.}

\BlankLine
\BlankLine

\If{$S$ splits $P_j$} {
$A \gets P_j - S$\;
$B \gets P_j - A$\;
}

\BlankLine

\lIf{$i < j$}{replace $P_j$ in $\mathcal{P}$ with $B,A$ in order}
\lElse{replace $P_j$ in $\mathcal{P}$ with $A,B$ in order}

\BlankLine

\KwRet{$\mathcal{P}$}\;

\BlankLine

\caption{$PivotPull(\mathcal{P},p,P,S)$} \label{alg:pivotPull}
\end{algorithm}

\begin{remark} \label{rem:pivotPull}
Suppose that $A,B$ are partition classes (in order) that replace partition class $P_j$ during algorithm~\ref{alg:pivotPush}.  Then either $A \subseteq S$ and $S \cap B = \emptyset$ (if $j < i$ in the algorithm) or $B \subseteq S$ and $A \cap S = \emptyset$ (if $i > j$ in the algorithm).  
\end{remark}

Algorithm~\ref{alg:linearExtension} below employs the new forms of refinement defined above to compute the desired linear extension.  As promised in section~\ref{sec:overview}, algorithm~\ref{alg:linearExtension} processes each slice in order, and for each, there are two rounds of refinement, corresponding below to the two inner loops.  

\begin{algorithm}

\BlankLine

\KwIn{The set of slices $\mathcal{S} = S_1,\ldots,S_n$ of a prime comparability graph $G$ defined by some LBFS ordering in which the initial vertex $x$ is a source.}

\KwOut{A linear extension of $G$.}

\BlankLine

$\mathcal{P} \gets$ the ordered partition $\{x\},V(G) - \{x\}$\;

\BlankLine

\For{$S_1$ \KwTo $S_n$ (in order)} {

\BlankLine

let $y$ be the initial vertex for $S_i$\;
let $P_y$ be the partition class in $\mathcal{P}$ (and throughout what follows) such that $y \in P_y$\;

\BlankLine

\While{there is a connected vertex $z$ for $S_i$ that has not been pivot and such that $z \in S_i - P_y$} {

\BlankLine

\uIf{$z \in \alpha_i(y)$} { 

$C \gets$ the co-component of $G[\alpha_i(y)]$ such that $z \in C$\;
$S \gets \alpha_i(z) \cup \alpha_i(y) - C$\;

}

\Else{

$C \gets \emptyset$\;
$S \gets \alpha_i(z)$\;

}

\BlankLine

$\mathcal{P} \gets PivotPull(\mathcal{P},z,P_y,C)$  \tcp{Algorithm~\ref{alg:pivotPull}}
$\mathcal{P} \gets PivotPush(\mathcal{P},z,P_y,S)$ \tcp{Algorithm~\ref{alg:pivotPush}}

\BlankLine

}

\BlankLine

replace $P_y$ in $\mathcal{P}$ with $\{y\},P_y - \{y\}$ in order

\BlankLine

\ForEach{connected vertex $z \in S_i$ (in order)} { 

$\mathcal{P} \gets GeneralizedPivot(\mathcal{P},z,\alpha_i(z))$ \tcp{Algorithm~\ref{alg:generalizedPivot}} \;

}

\BlankLine

}

\BlankLine
\KwRet{$\mathcal{P}$}\;

\caption{$LinearExtension(\mathcal{S})$} \label{alg:linearExtension}
\end{algorithm}

As alluded to in section~\ref{sec:overview}, the first round of refinement is designed to remove the neighbours of each initial vertex from its containing partition class.  Algorithms~\ref{alg:pivotPush} and~\ref{alg:pivotPull} are used for this purpose, the targeted partition class being the one containing the initial vertex.  During this round of refinement, the co-components are moved as one block for the reasons described in section~\ref{sec:coComponents}.  The ordered partition maintained throughout remains consistent with a linear extension.  

Once the initial vertex is isolated from its neighbours, it can be removed from its containing partition class with the resulting ordered partition still being consistent with a linear extension.  Each vertex is then processed in order and refinement takes place according to their active edges for that slice.  With the initial vertex isolated, and by proceeding in order, this has the effect of refining the ordered partition in terms of the maximal subslices of the current slice.  Intuitively, lemma~\ref{lem:sliceSameDir} ensures the result remains consistent with a linear extension.

To illustrate the operation of algorithm~\ref{alg:linearExtension}, we return to the graph in figure~\ref{fig:example} and the set of slices $S_1,\ldots,S_{10}$ defined earlier for it in section~\ref{sec:LBFS}.  It may also help to recall the corresponding slice tree in figure~\ref{fig:sliceTree} and the listing of active edges outlined in sections~\ref{sec:slices} and~\ref{sec:activeEdges}, respectively.  

In terms of algorithm~\ref{alg:linearExtension}, we first note that $\mathcal{P} = \{x\},\{b,y,u,z,q,w,r,v,a\}$ prior to any iteration of the outer for loop.  Now consider what happens on the first iteration of the outer for loop.  Note that each of $b,y,u,z,q,w,r,v,z,$ and $a$ are connected vertices for $S_1$ not in the same partition class as $x$, the initial vertex for $S_1$.  However, since the other endpoint of each of their active edges shares the same partition class as them, the inner for loop has no impact on $\mathcal{P}$.  Of course, neither does the line immediately preceding the inner for loop.  The following highlights the values of $\mathcal{P}$ following each iteration of the inner for loop:

\begin{itemize}

\item (After $x$ is pivot): $\mathcal{P} = \{x\},\{y,u,z,q,w,r,v,a\},\{b\}$;

\item (After $b$ is pivot): $\mathcal{P} = \{x\},\{y,u,z,q,w,r,v\},\{a\},\{b\}$;

\item (After $y,u,z,q,w,r,$ and $v$ are pivot): $\mathcal{P} = \{x\},\{y,u,z,q,w,r,v\},\{a\},\{b\}$;

\item (After $a$ is pivot): $\mathcal{P} = \{x\},\{z\},\{y,u,q,w,r,v\},\{a\},\{b\}$.

\end{itemize}

Now consider what happens on the second iteration of the outer for loop.  Since there are no active edges for $S_2$, there are no connected vertices for $S_2$, and this iteration has no impact on $\mathcal{P}$.  

On the third iteration of the outer for loop, we have the following connected vertices for $S_3$: $y,u,z,q,w,r,v$; of these, $z$ is the only one not in the same partition class as $y$, the initial vertex for $S_3$.  Therefore the while loop executes, and since $z \in \alpha_3(y)$, the first branch of the if statement succeeds.  Looking at $G[\alpha_3(y)]$ = $G[\{u,z,q,w\}]$, we notice that $C = \{z,q,w\}$ in this case; also note that $S = \{y,u,q,w\}$.  Therefore, $\mathcal{P} = \{x\},\{z\},\{q,w\},\{y,u,r,v\},\{a\},\{b\}$ after the first iteration of the while loop.

Notice now that $q$ and $w$ are no longer in the same partition class as $y$.  Therefore the while loop continues to iterate.  Assume that $q$ is selected on the next iteration.  Since $q \in \alpha_3(y)$, the first branch of the if statement succeeds.  Once again, $C = \{z,q,w\}$, but now $S = \{y,v,u,z,w\}$.  Therefore, $\mathcal{P} = \{x\},\{z\},\{q,w\},\{r,v\},\{y,u\},\{a\},\{b\}$ after this iteration of the while loop.  On the next iteration of the while loop, assume $w$ is selected.  Then $C = \{z,q,w\}$ and $S = y,u,z,q$.  As a result, there is no change, and we still have $\mathcal{P} = \{x\},\{z\},\{q,w\},\{r,v\},\{y,u\},\{a\},\{b\}$.

As a result of $q$ and $w$'s refinement on the previous iterations of the while loop, we no longer have $r$ and $v$ sharing the same partition class.  Therefore the while loop continues to operate.  Assuming $r$ is chosen next, the second branch of the if-statement succeeds because $r \notin \alpha_3(y)$.  Hence, $C = \emptyset$ and $S = \alpha_3(r) = \{u\}$.  Therefore, $\mathcal{P} = \{x\},\{z\},\{q,w\},\{r,v\},\{y\},\{u\},\{a\},\{b\}$ after this iteration.  Assuming now that $v$ is chosen on the next iteration, we have $C = \emptyset$ (because the second branch of the if statement succeeds) and $S = \{q\}$, meaning there is no change and $\mathcal{P} = \{x\},\{z\},\{q,w\},\{r,v\},\{y\},\{u\},\{a\},\{b\}$.

It is now no longer possible to further refine $y$'s partition class.  Thus, even though $u$ is a connected vertex not in $y$'s partition class, when it is chosen, which will have to be on the next iteration of the while loop, it will have no effect on $\mathcal{P}$.  After this, all connected vertices for $S_3$ not in $y$'s partition class will have been selected and no further iterations of the while loop will execute.  

Moving on, observe that the line immediately prior to the inner for loop also has no impact on $\mathcal{P}$ due to $y$ being in a singleton class.  We now consider what happens to $\mathcal{P}$ on each iteration of the inner for loop.  

Recall that the connected vertices for $S_3$ (in order) are: $y,u,z,q,w,r,v$.  The following demonstrates the state of $\mathcal{P}$ after each of these vertices is selected as pivot on subsequent iterations of the inner for loop:

\begin{itemize}

\item (After $y$ is pivot): $\mathcal{P} = \{x\},\{z\},\{q,w\},\{r,v\},\{y\},\{u\},\{a\},\{b\}$;

\item (After $u$ is pivot): $\mathcal{P} = \{x\},\{z\},\{q,w\},\{r\},\{v\},\{y\},\{u\},\{a\},\{b\}$;

\item (After $z,q,w$ and $r$ are pivot): $\mathcal{P} = \{x\},\{z\},\{q,w\},\{r\},\{v\},\{y\},\{u\},\{a\},\{b\}$;

\item (After $v$ is pivot): $\mathcal{P} = \{x\},\{z\},\{q\},\{w\},\{r\},\{v\},\{y\},\{u\},\{a\},\{b\}$.

\end{itemize}

Notice that all partition classes are now singleton sets.  As algorithm~\ref{alg:LBFS} clearly never unions any two partition classes, we can stop our tracing of the algorithm as $\mathcal{P}$ will not be further modified.  Observe that $\mathcal{P}$ defines a linear extension of $G$ corresponding to the directions on the edges in figure~\ref{fig:transitiveOrientation}.

\subsection{Implementation and Running-Time}

\subsubsection{Preprocessing}

First consider the portion of algorithm~\ref{alg:linearExtension} outside the outer for loop.  Recall from section~\ref{sec:refinement} that we can assume we have access to the underlying LBFS ordering defining the slices with $x$ as the initial vertex.  We can therefore assume access to $x$ in $O(n)$ time.  Given $x$, the initial partition $\{x\},V(G) - \{x\}$ can be formed in $O(n)$ time using the partition refinement data-structure from section~\ref{sec:partitionRefinement}.  Before moving on to the outer for loop itself, we will need to perform some preprocessing to ensure its efficient execution.  

First, we need to compute the initial vertex for each slice.  But recall the correspondence between slices and vertices.  The initial vertex for $S_i$  is therefore the $i^{th}$ vertex in the LBFS ordering.  Therefore the initial vertex for each slice can be computed in $O(n)$ time.  Given this, access to $P_y$ can be obtained in constant time throughout the outer for loop using a simple array data structure.    

Next, we will need to compute an ordered list of connected vertices for each slice, where the ordering is determined by the LBFS.  This will specifically be needed for the efficient execution of the inner for loop.  To help us, we'll use the graph induced by the active edges for each slice, which we computed earlier during initialization (see section~\ref{sec:refinement}); we'll assume this is represented as an adjacency list.  Notice that it suffices to order the vertices in this adjacency list according to the underlying LBFS ordering.

To do this, we'll need to first make a copy of the LBFS ordering, to which as we noted above, we can assume access.  Then, we'll need to associate a linked list (initialized as empty) with each vertex in that copy.  After doing this, we'll process each vertex in the graph induced by the active edges for $S_1$, appending a reference to itself in the list associated with it in the LBFS copy.  We'll then do the same for each of $S_2,\ldots,S_n$ in order.  Afterwards, the list associated with each vertex in the LBFS copy clearly references its position in the adjacency list of every graph induced by active edges (as long as it appears there).  In other words, the list corresponds to those slices where the vertex is connected.

The next step is to go through the copy of the LBFS in order, traversing the list associated with each vertex in order, following each reference, then removing that vertex from the graph's adjacency list only to then append it to that same adjacency list.  After doing this, it should be clear that the list of vertices for every graph induced by the active edges for each slice has been ordered according to the underlying LBFS ordering.  

The cost of doing this is clearly proportional to the sum (over all vertices) of the number of slices each vertex belongs to for which it is connected.  By definition, each vertex is connected for a slice if it has an incident active edge for that slice.  But each edge incident to a vertex is active for at most one slice, by remark~\ref{rem:activeOnce}.  Hence, the total cost to order the vertices in the graph induced by the active edges of each slice is $O(n+m)$.

\subsubsection{The Outer For Loop}

Now consider each iteration of the outer for-loop.  As we noted above, each initial vertex $y$ and its partition class $P_y$ can be located in constant-time.  Notice that access to $P_y$ is only needed a constant number of times on each iteration.  Hence, the total cost of this access over the entire algorithm is ~$O(n)$.

To evaluate the condition in the while loop, we will maintain a list of connected vertices for $S_i$ that have not been pivot, and we will partition this list into two halves: those that are in $P_y$ and those that are not.  Given such a list, the condition in the while loop can clearly be evaluated in constant time using the data structure of section~\ref{sec:partitionRefinement}.

Such a list will first need to be constructed prior to any iteration of the while loop.  This can be performed by scanning the list of vertices in the graph induced by the active edges for $S_i$, to which we have access as noted earlier.  At each vertex, it is a constant time operation to verify it is in $P_y$ assuming a pointer to $P_y$ and the data structure for partition classes from section~\ref{sec:partitionRefinement}.  The total cost of initializing this list for each $S_i$ over the entire algorithm is therefore $O(n+m)$, by applying remark~\ref{rem:activeOnce} in the same fashion we did above.  Maintaining this list of vertices (that has not been pivot) is clearly proportional to the time to maintain $\mathcal{P}$ itself, and hence we turn to that now.  

Recall from section~\ref{sec:refinement} that the co-components of each $S_i$ can be assumed to be computed.  Therefore, if the first branch of the if-statement in the while loop succeeds, the vertices in $C$ can be identified and removed from $P_y$ in $|C|$ time during the call to algorithm~\ref{alg:pivotPull}.  A key observation for the running time is that this only needs to be done once.  That is, once a vertex from $C$ has been pivot all vertices in $C$ will have been removed from $P_y$ already.  Therefore, on any iteration of the outer for loop, the vertices in $C$ only need to be processed once.  Of course, each of these vertices is adjacent to the initial vertex for the slice being processed, and each of the incident edges is active for that slice.  So once again, by remark~\ref{rem:activeOnce}, the total cost of algorithm~\ref{alg:pivotPull} is $O(n+m)$.

Consider now the total cost of identifying and removing $S$ from $P_y$  over the course of the algorithm.  Of course, $S$ can be computed and removed from $P_y$ by traversing $\alpha_i(z)$ and $\alpha_i(y)$, to which we have constant time access via the graph induced by the active edges for $S_i$, as noted earlier.  But that might potentially lead to traversing $\alpha_i(y)$ multiple times during the while loop, which would not be consistent with linear time overall.  

The solution is to process $\alpha_i(y)$ and $y$ as one ``block''.  What we mean by this is that the first time $\alpha_i(y)$ is processed, the portion in $P_y$ should be made children of the vertex $y$ in the partition refinement data structure of section~\ref{sec:partitionRefinement}.  Then, for the duration of the while loop, ignore the requirement to include $\alpha_i(y)$ as part of $S$.  This works because the only time when $\alpha_i(y)$ must be added to $S$ is when the first branch of the conditional in the while loop succeeds, which happens when $y \in \alpha_i(z)$: if the vertices of $\alpha_i(y) \cap P_y$ are children of $P_y$, then they will move with $y$ as part of $\alpha_i(z)$ as they should.  

It may happen however that some portion of $\alpha_i(z) \cap P_y$ needs to be included in $S$ when the second branch of the conditional succeeds.  In this case, the vertices in question can be removed from ``under'' $y$ as though $y$ were a nested partition class within $P_y$.  This adds at most a constant amount of work to each step in the partition refinement implementation of section~\ref{sec:partitionRefinement}.  At the end of the while loop, one more traversal of $\alpha_i(y)$ can restore $P_y$ to its original condition.  

Overall, then, the identification and removal of $S$ from $P_y$ results in $\alpha_i(y)$ being traversed at most twice.  The remaining work involves traversing $\alpha_i(z)$ on each iteration of the while loop.  This work is consistent with linear time overall $O(n+m)$ due to remark~\ref{rem:activeOnce} as above.  

Since the line immediately preceding the inner for loop is clearly a constant time operation, it only remains to consider the inner for loop itself.  Using the partition refinement implementation of section~\ref{sec:partitionRefinement}, the cost of each iteration is clearly proportional to $\alpha_i(z)$.  Each vertex is only pivot once, so overall, the cost of the inner for loop is clearly consistent with linear time, meaning $O(n+m)$.  From this and our work above, we can conclude the following: 

\begin{lemma} \label{lem:runningTime}
Algorithm~\ref{alg:linearExtension} can be implemented to run in $O(n+m)$ time.
\end{lemma}

\subsection{Correctness} \label{sec:correctness}

Having shown that algorithm~\ref{alg:linearExtension} can be implemented to run in linear-time, it only remains to prove its correctness.  To this end, notice that on each iteration of the outer for-loop, the initial vertex of the slice under consideration is isolated in its own partition class.  Thus, the ordered partition at the conclusion of the outer for-loop induces an ordering of the vertices.  It therefore suffices to prove that the ordered partition maintained throughout algorithm~\ref{alg:linearExtension} is consistent with a linear extension.  In order to accomplish this, the corresponding invariant needs to be strengthened as below.

\begin{invariant} \label{inv:consistency}
After the $i^{th}$ iteration of the outer for-loop:

\begin{enumerate}

\item The endpoints of any edge that is active for the slice $S_i$ reside in different partition classes in $\mathcal{P}$, and if $i > 0$, each endpoint of such an edge is universal to the partition class containing the other endpoint;

\item If $i < n$, then there is at least one active edge with respect to $S_{i+1}$ whose endpoints reside in different partition classes;

\item The ordered partition $\mathcal{P}$ is consistent with a linear extension.

\end{enumerate}

\end{invariant}

The invariant provides some intuition as to the operation and correctness of the algorithm.  Condition~1 says that at the end of each iteration of the outer for-loop, no further refinement is possible using active edges with respect to already considered slices.  Moreover, condition~3 says that all refinement up until that point has been consistent with a linear extension.  And finally, condition~2 says that on the next iteration, there will be at least one ``seed'' edge to begin the process of refining using the active edges with respect to the next slice under consideration.  

Starting from this seed edge, refinement takes place in a manner and order that ensures that all active edges for that slice will be used to refine the ordered partition maintained by the algorithm, and in such a way that is consistent with a linear extension.  By the end of the algorithm, every edge will have been used to refine the ordered partition.  All such refinement will have been consistent with a linear extension.  As remarked above, only singleton sets will remain, thereby inducing a linear extension.

The rest of this section constitutes one large induction proof of invariant~\ref{inv:consistency}.  Each part -- base case, induction hypothesis, induction step -- will be stated and proved separately.  All statements should be interpreted in the context of algorithm~\ref{alg:linearExtension}; that is, $G$, $S_1,\ldots,S_n$, $x$, $y$, $P_y$, $C$, $S$, etc. have the meaning assigned to them in the algorithm.  

\begin{baseCase}
Invariant~\ref{inv:consistency} holds before any iteration of the outer for-loop.
\end{baseCase}

\begin{proof}
The initial ordered partition $\{x\},V(G) - \{x\}$ is consistent with a linear extension by virtue of $x$ being a source vertex.  Since the input graph $G$ is prime, it must be connected, and therefore $x$ must have at least one incident active edge with respect to the slice $S_1$, the other endpoint of which is obviously in a different partition class than is $x$.  The remaining condition of invariant~\ref{inv:consistency} holds vacuously.
\end{proof}

\begin{inductionHypothesis} 
Consider the $i^{th}$ iteration of the outer for-loop.  Then invariant~\ref{inv:consistency} holds after all prior iterations.  
\end{inductionHypothesis}

\paragraph{Induction Step} This will be comprised of three separate lemmas, one for each condition of invariant~\ref{inv:consistency}.  

\begin{lemma}
Assuming the induction hypothesis, condition~1 of invariant~\ref{inv:consistency} holds after the $i^{th}$ iteration of the outer for-loop.
\end{lemma}

\begin{proof}
To prove the first part of condition~1 of invariant~\ref{inv:consistency} we will instead prove the following claim: immediately prior to the $j^{th}$ iteration of the inner for-loop, all active edges with respect to $S_i$ that are incident to the first $j+1$ vertices in $S_i$ are such that their endpoints reside in different partition classes in $\mathcal{P}$.  

This cearly holds prior to any iteration of the inner for loop because the initial vertex for $S_i$ resides in its own partition class prior to any iteration of the inner for-loop.  So consider the $j^{th}$ iteration of the inner for-loop and assume the claim holds after all prior iterations.  

Let $z$ be the $j+1^{st}$ vertex in $S_i$ and choose $P \in \mathcal{P}$ be such that $z \in P$.  Suppose for contradiction that there exists a vertex $v \in \alpha_i(z)$ such that $v \in P$ as well.  Notice that $v$ cannot be amongst the first $j+1$ vertices in $S_i$.  Given the base case, it will be assumed that $z$ is not the initial vertex for $S_i$.  Therefore, by condition~3 of remark~\ref{rem:maximalSlice}, there is a vertex $u$ that is earlier than $z$ in $S_i$ such that $u$ is adjacent to $z$ but not adjacent to $v$, and in this context, $z \in \alpha_i(u)$ while $v \notin \alpha_i(u)$.  

Suppose that $u$ is the $\ell^{th}$ vertex in $S_i$.  Then $\ell \le j$.  So by the induction hypothesis, immediately prior to the $\ell - 1^{st}$ iteration of the inner for-loop, $u$ and $z$ resided in different partition classes in $\mathcal{P}$.  However, $v$ and $z$ must reside in the same partition class at this point as partition classes are clearly never merged during the inner for-loop.  So by remark~\ref{rem:generalizedPivot}, immediately prior to the $\ell^{th}$ iteration of the inner for-loop, $v$ and $z$ must reside in different partition classes in $\mathcal{P}$.  The desired contradiction follows by once again noting that no union ever occurs between a subset of one partition class and a subset of another partition class during the inner for-loop.  

This proves the claim and implies condition~1 of invariant~\ref{inv:consistency}.  Note that the second part of that condition follows as a corollary to the claim by applying remark~\ref{rem:generalizedPivot} to every non-vacuous iteration of the inner for loop.
\end{proof}

\begin{lemma} \label{lem:invariantItem2}
Assuming the induction hypothesis, condition~2 of invariant~\ref{inv:consistency} holds after the $i^{th}$ iteration of the outer for-loop.
\end{lemma}

\begin{proof}
If there is no active edge for $S_{i+1}$ then the lemma follows vacuously.  So assume that at least one such edge exists.  Note that both its endpoints are connected.  Now assume for contradiction that after the $i^{th}$ iteration of the outer for-loop, the endpoints of all such active edges reside in the same partition classes in $\mathcal{P}$.  Let $C'$ be the set of connected vertices for $S_{i+1}$.  Note that $C'$ is non-trivial by assumption.  Furthermore, there is a partition class $P \in \mathcal{P}$ such that $C' \subseteq P$.  Notice as well that $C'$ is a non-trivial module in $G[S_{i+1}]$, by remark~\ref{rem:disconnected}.  But since $G$ is prime, it is not a non-trivial module for $G$, hence, there is a vertex $q \notin S_{i+1}$ and vertices $u,v \in C'$ such that $q$ is adjacent to precisely one of them, say $u$.  Let $S_j$ be the least common ancestor of $q$ and $u$ in the slice tree.  Note that $j < {i+1}$, by item~1 of remark~\ref{rem:sliceTreeAncestor}.  Note as well that $uq$ is an active edge for $S_j$, by choice of $S_j$.  So by condition~1 of invariant~\ref{inv:consistency} and our induction hypothesis, $u$ and $q$ reside in different partition classes at the end of the $j^{th}$ iteration of the outer for loop.  Note that at no point during algorithm~\ref{alg:linearExtension} is there any union between some subset of one partition class and a subset of a different partition class.  Therefore $u$ and $v$ are in the same partition class after the $j^{th}$ iteration of the outer for loop.  So by condition~1 of invariant~\ref{inv:consistency} and our induction hypothesis, $q$ must be adjacent to $v$, a contradiction.  
\end{proof}

\begin{lemma}
Assuming the induction hypothesis, condition~3 of invariant~\ref{inv:consistency} holds after the $i^{th}$ iteration of the outer for-loop.
\end{lemma}

\begin{proof}
Consider the $i+1^{st}$ iteration of the outer for loop.  We will prove that $\mathcal{P}$ is consistent with a linear extension after this iteration by proving the following three claims: 

\begin{enumerate}

\item $\mathcal{P}$ is consistent with a linear extension after every iteration of the while loop;  

\item $N(y) \cap P_y = \emptyset$ after the while loop has completed;

\item $\mathcal{P}$ is consistent with a linear extension after every iteration of the inner for loop.

\end{enumerate}

\paragraph{Claim 1:} Notice that $\mathcal{P}$ is consistent with a linear extension prior to any iteration of the while loop by our induction hypothesis.  So assume $\mathcal{P}$ is consistent with a linear extension after some number of iterations of the while loop.  Note that on any iteration of the while loop, $\mathcal{P}$ only changes by the splitting of $P_y$ as a result of calls to algorithm~\ref{alg:pivotPull} and algorithm~\ref{alg:pivotPush}.  

After the call to algorithm~\ref{alg:pivotPull}, $P$ is replaced in $\mathcal{P}$ with $P_y - C$,$C \cap P_y$ (not necessarily in that order).  Observe that if $\mathcal{P}$ is no longer consistent with a linear extension, it is because of an edge with one endpoint in $P_y - C$ and the other in $C \cap P_y$.  However, by lemma~\ref{lem:coSameDir}, these edges must share the same direction as those between $z$ and $C$, and therefore $\mathcal{P}$ remains consistent with a linear extension after the call to algorithm~\ref{alg:pivotPull}.  

For the call to algorithm~\ref{alg:pivotPush}, $P_y$ is replaced with $P_y - S$ and $P_y \cap S$ (not necessarily in that order).  Observe that $z$ is universal to $P_y \cap S$.  We argue that all edges between $P_y - S$ and $P_y \cap S$ must be directed the same as those between $P_y - \cap S$ and $z$.  Suppose for contradiction there is an edge that is not, between say $u \in P_y - S$ and $v \in P_y \cap S$.  If $z$ is not adjacent to $u$, then the contradiction follows immediately.  Therefore $u$ and $z$ must be adjacent.  

We now have two cases.  First, assume that $z \in \alpha_{i+1}(y)$.  Then $u \notin S_{i+1}$, because if it were, then either $u \in \alpha_{i+1}(y)$ or $u \in \alpha_{i+1}(z)$, neither of which hold by definition of $S$ in this case.  So let $S_j$ be the least common ancestor of $u$ and $v$ in the slice-tree.  Notice that $j < i+1$ and the edge $uv$ is active for $S_j$.  Also notice that since algorithm~\ref{alg:linearExtension} clearly never unions a subset of one partition class with a subset of another partition class, and since $u,v \in P_y$ prior to the call to algorithm~\ref{alg:pivotPush}, it follows that $u$ and $v$ were in the same partition class at the conclusion of the $j^{th}$ iteration of the outer for loop.  But this contradicts condition~1 of invariant~\ref{inv:consistency} and our induction hypothesis.

So now assume that $z \notin \alpha_{i+1}(y)$.  Then $S = \alpha_{i+1}(z)$, meaning $z$ and $v$ are in different maximal sub slices of $S_{i+1}$.  Furthermore, $u \notin \alpha_{i+1}(z)$, meaning either $u \notin S_{i+1}$ or there is a maximal subslice of $S_{i+1}$, call it $S$, such that $u,z \in S$.  A similar argument as in the previous case rules out the former possibility.  However, the latter case is ruled out by lemma~\ref{lem:sameDirection}.     

\paragraph{Claim 2:}  Let $M$ be the set defined as the union of $\{y\}$, $N(y) \cap P_y \cap S_{i+1}$, and all maximal subsets of $S_{i+1}$ that are subsets of $P_y$.  We will show that $M$ is a module for $G[S_{i+1}]$.  Consider any vertex $q \in S_{i+1} - P_y$.  We have two cases, depending on whether $q$ is adjacent to $y$.  First consider the case where $q$ is adjacent to $y$.  Then by definition, $q \in \alpha_{i+1}(y)$, and so $q$ is connected for $S_{i+1}$.  As $q \notin P_y$, we know that $q$ was pivot on some iteration of the while loop.  So by remark~\ref{rem:pivotPush} and the definition of $S$ in algorithm~\ref{alg:linearExtension} in this case, $q$ is universal to $P_y$ and hence universal to $M$.  

Now consider the case where $q$ is not adjacent to $y$.  So in this case $q \notin \alpha_{i+1}(y)$.  Let $S'$ be the maximal subslice of $S_{i+1}$ such that $q \in S'$.  Note that there is no vertex $u \in S' \cap M$ by definition of $M$ and since $q \notin P_y$.  Also note that if $q$ is disconnected, then so are all vertices in $S'$, and in this case, all vertices in $S'$ are isolated from $M$, by definition.  So assume that $q$ is connected.  Then it was pivot on some iteration of the while loop, since $q \notin P_y$.  So by remark~\ref{rem:pivotPush} and the definition of $S$ in algorithm\ref{alg:linearExtension} in this case,  we have that $q$ is isolated from $P_y$ and hence $M$.  We will now show that all vertices in $(S' \cap P_y) - M$ are similarly isolated from $M$ to conclude that $M$ is a module for $G[S_{i+}$.

Consider some vertex $r \in (S' \cap P_y) - M$.  Assume for contradiction that $r$ is adjacent to some vertex $u \in M$.  Notice that $u$ cannot appear in an earlier maximal subslice than $S'$ due to item~2 of remark~\ref{rem:maximalSlice}.  So $u$ appears in a later maximal subslice than $S'$, call it $S''$.  Then by item~3 of remark~\ref{rem:maximalSlice}, there is a vertex $v \notin S' \cup S''$ that is universal to $S'$ and isolated from $S''$.  So $q,r \in \alpha_{i+1}(v)$ but $u \notin \alpha_{i+1}(v)$, and $v$ is connected for $S_{i+1}$.  And recall that $q$ is isolated from $P_y$.  Therefore $v \notin P_y$.  Hence, $q$ was pivot on some iteration of the while loop.  Since $S''$ appears later than $S'$, and $q \notin \alpha_{i+1}(y)$, we also have $u \notin \alpha_{i+1}(y)$.  Thus, after the iteration on which $v$ was pivot, we cannot have $u,r \in P_y$ as we do, by remark~\ref{rem:pivotPush}.

From above, it follows that $M$ is a module for $G[S_{i+1}]$.  Applying a similar argument as in lemma~\ref{lem:invariantItem2}, we can conclude that $M$ is a module for $G$.  However, since $G$ is prime, this module must be trivial.  But we know $y \in M$, by definition, so it must be that $N(y) \cap P_y \cap S_{i+1} = \emptyset$.  In other words, no neighbour of $y$ in $S_{i+1}$ appears in $P_y$.  We conclude the claim by citing item~1 of invariant~\ref{inv:consistency} and lemma~\ref{lem:activeEdgePartition}.

\paragraph{Claim 3:} By claim 2 above, $\mathcal{P}$ is consistent with a linear extension prior to the first iteration of the inner for loop.  So assume it is consistent with a linear extension after some number of iterations of the inner for loop.  Consider the pivot $z$ as defined by the next iteration of the inner for loop.  Suppose that during this iteration, $z$ splits a partition class, replacing it with $A,B$, in order according to algorithm~\ref{alg:generalizedPivot}.  Note that $z \notin A \cup B$.  Assume that $z$ appears in a partition class earlier than $A$ in $\mathcal{P}$ (the other case is symmetric).  It suffices to show that all edges between $A$ and $B$ are directed toward $B$.  

Let $uv$ be such an edge, and say $u \in A$ and $b \in B$.  It follows that $u \notin \alpha_{i+1}(z)$ and $v \in \alpha_{i+1}(z)$.  If $u \notin N(z)$, then the edge $uv$ must clearly be directed toward $v$ as it is here.  So assume that $u \in N(z)$.  Then by condition~1 of invariant~\ref{inv:consistency} along with the induction hypothesis, it follows that the edge $uz$ is not active for any $S_j, j \le i + i$.  Therefore $u$ and $z$ must be in the same maximal subslice for $S_{i+1}$.  By lemma~\ref{lem:sameDirection}, it suffices to consider the case where $u,z \in N(y)$.  Here we have $u,z \in \alpha_{i+1}(y)$ as well.  

Observe that $v \neq y$ because on the prior iteration of the inner for loop, both $u$ and $v$ resided in the same partition class, and clearly $y$ resides in its own singleton partition class throughout the inner for loop.  Since $v \in \alpha_{i+1}(z)$, it follows that $v \notin N(y)$ in this case, and therefore $v \notin \alpha_{i+1}(y)$.  It follows that on the first iteration of the inner for loop, $y$ would have split the partition class containing $u$ and $v$, contradicting our assumption regarding $A$ and $B$.  It follows that $u,v \notin N(y)$.  Therefore $\mathcal{P}$ is consistent with a linear extension after each iteration of the inner for loop.  
\end{proof}

\section{Conclusion} \label{sec:conclusion}

This paper describes a linear-time algorithm to compute a linear extension of a prime comparability graph.  It does so by extending the elegant partition refinement framework of~\cite{m.mcconnell:linear-time}.  Some initialization is required as part of its implementation, but this work only requires tree and list traversals.  This represents a significant simplification over the only other linear-time algorithm for the problem due to~\cite{m.mcconnell:linearTransitive}.  It was achieved by using LBFS to accomplish the same effect as modular decomposition in the earlier algorithm.  

The main problem with that earlier algorithm was its reliance on the notoriously difficult linear-time modular decomposition algorithm from that same paper.  A simpler, linear-time alternative has since been developed in~\cite{tedder:modular} and could likely be used in its place, but it would not accomplish the same degree of simplification as observed here with LBFS.  Nevertheless, a reduction from modular decomposition is still required to compute a linear extension for non-prime comparability graphs, and for that purpose, the simpler, linear-time modular decomposition algorithm of~\cite{tedder:modular} is the preferred choice.  It can be viewed as the companion to this paper.  Many of its insights contributed to the development of the algorithm presented here.

It is important to draw the distinction between the transitive orientation problem and the \emph{comparability graph recognition problem}.  As already stated, the transitive orientation problem asks you to orient the edges of a comparability graph.  In that case, you are given a comparability graph.  In contrast, the recognition problem gives you an arbitrary graph and asks you to determine if it is a comparability graph.  One solution would be to apply the transitive orientation algorithm from this paper and then verify if the resulting orientation is actually transitive.  Rather surprisingly, this verification is the harder problem, at least in terms of running-time.  The book~\cite{spinrad:book} has a good discussion on this.  As summarized there, the problem has been shown to be equivalent to finding a triangle in a graph, and in both cases, the fastest algorithm requires matrix multiplication, whose complexity is currently~$O(n^{2.373})$~\cite{davie:matrix}.  Some kind of hardness result in terms of matrix multiplication or any more efficient algorithm would be a significant result.

Fortunately, a transitive orientation of a comparability graph is frequently all that is required for many applications~\cite{m.mcconnell:linear-time}.  Other applications involving co-comparability graphs and permutation graphs require a transitive orientation of the complement graph.  One of these applications is permutation graph recognition.  For efficiency reasons, it is preferable if an orientation of the complement can be computed without explicitly computing that complement.  This is possible by implicitly representing such an orientation by a linear extension.  With minor modifications, the previous linear-time transitive orientation algorithm of~\cite{m.mcconnell:linearTransitive} can do so in linear-time.  The algorithm developed in this paper is capable of the same.  Among other implications, this would result in a simpler linear-time permutation graph recognition algorithm than the one presented in~\cite{m.mcconnell:linearTransitive}, which is the only other to date.  

Two changes are required to transform the algorithm developed in this paper so that it computes a transitive orientation of the complement graph.  First, as in~\cite{habib:stacs}, the order of the subdivided classes in algorithm~\ref{alg:pivot} would need to be reversed.  Doing so has the effect of computing a source vertex for the complement while also maintaining consistency with a linear extension of the complement throughout refinement.  Moreover,~\cite{habib:stacs} showed that the same change results in an LBFS ordering of the complement when applied to the partition refinement implementation of LBFS.  The second modification concerns the special role played by the co-components of slices.  This would now have to be played by the components of slices.  But to avoid ``touching'' disconnected vertices and thereby maintain linear-time, rather than co-components being ``pulled'' out of $P_y$ by their other halves during the first round of refinement, components will have to be ``pushed out'' from within $P_y$ to join their other halves.  The proof of correctness and running-time would be similar to the one in this paper.  

\section{Acknolwedgements}

The author wishes to thank Professor Derek Corneil of the University of Toronto for suggesting the problem and providing guidance and supervision during the research and preparation of this paper.  This paper bears his imprint, and is the result of many fruitful discussions between him and the author.  The author especially wishes to thank Professor Corneil for his unwavering support throughout the preparation of this paper.  He has been instrumental in ensuring the eventual presentation of this work here.

\bibliographystyle{plain}
\bibliography{TransitiveOrientation}

\end{document}